\documentclass[conference]{IEEEtran}
\usepackage{graphicx}
\usepackage{amsmath}
\usepackage{amssymb}
\usepackage{tikz}
\usepackage{amsthm}

\newcommand{\up}{\vspace{-0.2cm}}
\newtheorem{definition}{Definition}
\newtheorem{theorem}{Theorem}
\newtheorem{proposition}{Proposition}
\newtheorem{example}{Example}

\ifCLASSINFOpdf
\else
\fi
\begin{document}
%
\title{Zero-one laws for provability logic: Axiomatizing\\ validity in almost all models and almost all frames}


\author{\IEEEauthorblockN{Rineke Verbrugge\\Department of Artificial Intelligence, 
University of Groningen, 
e-mail L.C.Verbrugge@rug.nl}
\IEEEauthorblockA{ }
}


%


\IEEEoverridecommandlockouts
\IEEEpubid{\makebox[\columnwidth]{978-1-6654-4895-6/21/\$31.00~
\copyright2021 IEEE \hfill} \hspace{\columnsep}\makebox[\columnwidth]{ }}
\maketitle

\begin{abstract}
It has been shown in the late 1960s that each formula of first-order logic without constants and function symbols obeys a zero-one law: As the number of elements of finite models increases, every formula holds either in almost all or in almost no models of that size. Therefore, many properties of models, such as having an even number of elements, cannot be expressed in the language of first-order logic. 
For modal logics, limit behavior for models and frames may differ. Halpern and Kapron proved zero-one laws for classes of models corresponding to the modal logics K, T, S4, and S5. They also proposed zero-one laws for the corresponding classes of frames, but their zero-one law for K-frames has since been disproved.

In this paper, we prove zero-one laws for provability logic with respect to both model and frame validity. Moreover, we axiomatize validity in almost all irreflexive transitive finite models and in almost all irreflexive transitive finite frames, leading to two different axiom systems. In the proofs, we use a combinatorial result by Kleitman and Rothschild about the structure of almost all finite partial orders. On the way, we also show that a previous result by Halpern and Kapron about the axiomatization of almost sure frame validity for S4 is not correct.  Finally, we consider the complexity of deciding whether a given formula is almost surely valid in the relevant finite models and frames.
\end{abstract}


%
\IEEEpeerreviewmaketitle

\section{Introduction}

In the late 1960s, Glebskii and colleagues proved that first-order logic without function symbols satisfies a zero-one law, that is, every formula is either almost always true or almost always false in finite models \cite{glebskii1969}. More formally, let $L$ be a language of first-order logic and let $A_n(L)$ be the set of all {\em labelled} $L$-models with universe $\{1, \ldots, n\}$. Now let $\mu_n(\sigma)$ be the fraction of members of $A_n(L)$ in which $\sigma$ is true, i.e.,
 \[\mu_n(\sigma) = \frac{ \mid \{ M \in A_n(L) : M \models \sigma \}\mid}{ \mid A_n(L)\mid}\] 
 Then for every $\sigma \in L$, $\lim_{n\to\infty} \mu_n(\sigma) = 1$ or
$\lim_{n\to\infty} \mu_n(\varphi) = 0$.\footnote{The distinction between labelled and unlabelled probabilities was introduced by Compton~\cite{compton1987}. The unlabelled count function counts the number of isomorphism types of size $n$, while the labelled count function counts the number of labelled structures of size $n$, that is, the number of relevant structures on the universe $\{1,\ldots,n\}$. It has been proved both for the general zero-one law and for partial orders that in the limit, the distinction between labelled and unlabelled probabilities does not make a difference for zero-one laws~\cite{fagin1976,compton1987,compton1988b}. Per finite size $n$, labelled probabilities are easier to work with than unlabelled ones~\cite{goranko2003}, so we will use them in the rest of the article.}

This was also proved later but independently by Fagin~\cite{fagin1976}; Carnap had already proved the zero-one law for first-order languages with only unary predicate symbols~\cite{Carnap1950} (see \cite{compton1988,goranko2003} for  nice historical overviews of zero-one laws).
 Later, Kaufmann showed that monadic existential second-order logic does not satisfy a zero-one law~\cite{kaufmann1987}. Kolaitis and Vardi have made the border more precise by showing that a zero-one law  holds for the fragment of existential second-order logic  ($\Sigma^1_1$) in which the first-order part of the formula belongs to the Bernays-Sch\"{o}nfinkel class ($\exists^\ast \forall^\ast$ prefix) or the Ackermann class ($\exists^\ast \forall\exists^\ast$ prefix)~\cite{kolaitis1987,kolaitis1990}; however, no zero-one law holds for any other class, for example, the G\"{o}del class ($ \forall^2\exists^\ast$ prefix)~\cite{pacholski1989}. 
 Kolaitis and Vardi proved that a zero-one law does hold for the infinitary finite-variable logic $\mathcal{L}^{\omega}_{\infty\omega}$, which implies that a zero-one law also holds for LFP(FO), the extension of first-order logic with a least fixed-point operator~\cite{kolaitis1992}.
 
The above zero-one laws and other limit laws have found applications in database theory~\cite{gurevich1993,halpern2006,libkin2013} and algebra~\cite{zaid2017}. In AI, there has been great interest in asymptotic conditional probabilities and their relation to default reasoning and degrees of belief~\cite{grove1996,halpern2006}. 

In this article, we focus on zero-one laws for a modal logic that imposes structural restrictions on its models, namely, provability logic, which is sound and complete with respect to finite strict (irreflexive) partial orders~\cite{segerberg1971}.

The zero-one law for first-order logic also holds when restricted to partial orders, both reflexive and irreflexive ones, as proved by Compton \cite{compton1988b}. To prove this, he used a surprising combinatorial result by Kleitman and Rothschild~\cite{kleitman1975} on which we will also rely for our results. Let us give a 
summary.\vspace{-0.08cm}



\subsection{Kleitman and Rothschild's result on finite partial orders}
\label{KR}

Kleitman and Rothschild proved that with asymptotic probability 1, finite partial orders have a very special structure: There are no chains  $u<v<w<z$ of more than three objects 
and the structure can be divided into three levels:\vspace{-0.13cm}
\begin{itemize}
\item $L_1$, the set of minimal elements;
\item $L_2$, the set of elements immediately succeeding elements in $L_1$;
\item $L_3$, the set of elements immediately succeeding elements in $L_2$.\vspace{-0.13cm}
\end{itemize}
Moreover, the ratio of the expected size of $L_1$ to $n$ and of the expected size of $L_3$ to $n$ are both $\frac{1}{4}$, while the ratio of the expected size of $L_2$ to $n$ is $\frac{1}{2}$.
As $n$ increases, each element in $L_1$ has as immediate successors asymptotically half of the elements of $L_2$ and each element in $L_3$ has as immediate predecessors asymptotically  half of the elements of $L_2$~\cite{kleitman1975}.\footnote{Interestingly, it was recently found experimentally that for smaller $n$ there are strong oscillations, while the behavior stabilizes only around $n=45$~\cite{Henson2017}.}  Kleitman and Rothschild's theorem holds both for reflexive (non-strict) and for irreflexive (strict) partial orders. 



\subsection{Zero-one laws for modal logics: Almost sure model validity}
In order to describe the known results about zero-one laws for modal logics with respect to the relevant classes of models and frames, we first give reminders of some well-known definitions and results.

\noindent
Let $\Phi=\{p_1,\ldots, p_k\}$ be a finite set of propositional atoms\footnote{In the rest of this paper in the parts on almost sure model validity, we take $\Phi$ to be finite, although the results can be extended to enumerably infinite $\Phi$ by the methods described in~\cite{halpern1994,grove1996}.} and let $L(\Phi)$ be the modal language over $\Phi$, inductively defined as the smallest set closed under:
\begin{enumerate}
\item If $p \in \Phi$, then $p\in L(\Phi)$.
\item If $A\in L(\Phi)$ and $B \in L(\Phi)$, then also $\neg A \in L(\Phi)$, $\Box A \in L(\Phi)$, $\Diamond(\varphi)\in L(\Phi)$, $(A \wedge B) \in L(\Phi)$, $(A \vee B)\in L(\Phi)$, and $(A \rightarrow B) \in L(\Phi)$.
\end{enumerate}

\noindent
A {\em Kripke frame} (henceforth: frame) is a pair $F=(W, R)$ where $W$ is a non-empty set of worlds and $R$ is a binary accessibility relation. A {\em Kripke model} (henceforth: model) $M=(W, R, V)$ consists of a frame $(W, R)$ and a valuation function $V$ that assigns to each atomic proposition in each world a truth value $V_w(p)$, which can be either 0 or 1. The truth definition is as usual in modal logic, including the clause:\vspace{-0.1cm}
\[M,w \models \Box \varphi \mbox{ if and only if}\]
\vspace{-0.6cm}
\[\mbox{for all } w' \mbox{ such that } wRw', M,w'\models \varphi.\vspace{-0.1cm}\]

\noindent
A formula $\varphi$ is valid in model  $M=(W, R, V)$ (notation $M \models \varphi$) iff for all $w\in W$, $M,w\models\varphi$. 

\noindent
A formula $\varphi$ is valid in  frame $F=(W, R)$ (notation $F\models\varphi$) iff for all valuations $V$, $\varphi$ is valid in the model $(W, R, V)$.\\
\noindent
Let $\mathcal{M}_{n,\Phi}$ be the set of finite models over $\Phi$ with set of worlds $W=\{1, \ldots, n\}$. We take $\nu_{n,\Phi}$ to be the uniform probability distribution on $\mathcal{M}_{n,\Phi}$. Let $\nu_{n,\Phi}(\varphi)$ be the measure in $\mathcal{M}_{n,\Phi}$ of the set of models in which $\varphi$ is valid.\\
\noindent
Let $\mathcal{F}_{n,\Phi}$ be the set of finite frames with set of worlds $W=\{1, \ldots, n\}$. We take $\mu_{n,\Phi}$ to be the uniform probability distribution on $\mathcal{F}_n$. Let $\mu_{n,\Phi}(\varphi)$ be the measure in $\mathcal{F}_n$ of the set of frames in which $\varphi$ is valid. \\

\noindent
Halpern and Kapron proved that every formula $\varphi$ in modal language $L(\Phi)$ is either valid in almost all models  (``almost surely true'') or not valid in almost all models (``almost surely false'')~\cite[Corollary 4.2]{halpern1994}: 
\[ \mbox{Either } \lim_{n\to\infty} \nu_{n,\Phi}(\varphi) = 0 \mbox{ or } \lim_{n\to\infty} \nu_{n,\Phi}(\varphi) = 1.\]

\noindent
In fact, this zero-one law for models already follows from the zero-one law for first-order logic~\cite{glebskii1969,fagin1976} by Van Benthem's translation method \cite{Benthem1976,Benthem1983}.\label{Benthem} As reminder, let $^\ast$ be given by: 
\begin{itemize}
\item $p_i^\ast = P_i(x)$ for atomic sentences $p_i\in \Phi$;
\item $(\neg \varphi)^\ast = \neg \varphi^\ast$;
\item $(\varphi \wedge \psi)^\ast = (\varphi^\ast \wedge \psi^\ast)$ (
similar for other binary operators);
\item $(\Box \varphi)^\ast = \forall y(Rxy \rightarrow \varphi^\ast [y/x])$.
\end{itemize}
Van Benthem mapped each model $M=(W,R,V)$ to a classical model $M^\ast$ with as objects the worlds in $W$ and the obvious binary relation $R$, while for each atom $p_i\in\Phi$, $P_i=\{ w \in W \mid M,w\models p_i \}= \{ w \in W \mid V_w(p_i)=1 \}$. 
Van Benthem then proved that for all $\varphi \in L(\Phi)$, $M\models \varphi$ iff $M^\ast \models \forall x \; \varphi^\ast$~\cite{Benthem1983}. Halpern and Kapron~\cite{halpern1992,halpern1994}  showed that a zero-one law for modal models immediately follows by Van Benthem's result and the zero-one law for first-order logic. 

By Compton's above-mentioned result that the zero-one law for first-order logic holds when restricted to the partial orders~\cite{compton1988b}, this modal zero-one law can also be restricted to finite models on reflexive or irreflexive partial orders, so that a zero-one law for finite models of provability logic immediately follows. However, one would like to prove a stronger result and axiomatize the set of formulas $\varphi$ for which $\lim_{n\to\infty} \nu_{n,\Phi}(\varphi) = 1$. Also, Van Benthem's result does not allow proving zero-one laws for classes of frames instead of models: We have $F\models \varphi$ iff $F^\ast \models \forall P_1 \ldots \forall P_n \forall x \varphi^\ast$, but the latter formula is not necessarily a negation of a formula in $\Sigma^1_1$ with its first-order part in one of the Bernays-Sch\"{o}nfinkel or Ackermann classes (see~\cite{halpern1994}).

Halpern and Kapron~\cite{halpern1992,halpern1994} aimed to fill in the above-mentioned gaps for the modal logics {\bf K},  {\bf T}, {\bf S4} and {\bf S5} (see~\cite{chellas1980} for definitions). They proved zero-one laws for the relevant classes of finite models for these logics. 
For all four, they axiomatized the classes of sentences that are almost surely true in the relevant finite models.

\subsection{The quest for zero-one laws for frame validity}
Halpern and Kapron's paper also contains descriptions of four zero-one laws with respect to the classes of finite frames corresponding to {\bf K},  {\bf T}, {\bf S4} and {\bf S5}.~\cite[Theorem 5.1 and Theorem 5.15]{halpern1994}:
Either $\lim_{n\to\infty} \mu_{n,\Phi}(\varphi) = 0$ or $\lim_{n\to\infty} \mu_{n,\Phi}(\varphi) = 1$.\\
They proposed  four axiomatizations for the sets of formulas that would be  almost always valid in the corresponding four classes of frames~\cite{halpern1994}. However, almost 10 years later, Le Bars surprisingly proved them wrong with respect to the zero-one law for {\bf K}-frames~\cite{Bars2002}. By proving that the formula $q \wedge \neg p \wedge \Box\Box((p\vee q) \rightarrow \neg\Diamond(p\vee q)) \wedge \Box\Diamond p$ does {\em not} have an asymptotic probability, he showed that in fact {\em no} zero-one law holds with respect to all finite frames. Doubt had already been cast on the zero-one law for frame validity by Goranko and Kapron, who proved 
that the formula $\neg \Box\Box(p \leftrightarrow \neg \Diamond p)$ fails in the countably infinite random frame, while it is almost surely valid in {\bf K}-frames~\cite{goranko2003}. (See also
~\cite[Section 9.5]{Goranko2007}).\footnote{We will show in this  paper that for irreflexive partial orders, almost-sure frame validity in the finite {\em does} transfer to validity in the corresponding countable random Kleitman-Rothschild frame, and that the validities are quite different from those for almost all {\bf K} frames (see Section~\ref{Frames}).}
Currently, the problem of axiomatizing the modal logic of almost sure frame validities for finite {\bf K}-frames appears to be open.\footnote{For up to 2006: see~\cite{Goranko2007}; for more recently:~\cite{Goranko2020} .}

As a reaction to Le Bars' counter-example, Halpern and Kapron~\cite{halpern2003erratum} published an erratum, in which they showed exactly where their erstwhile proof of~\cite[Theorem 5.1]{halpern1994} had gone wrong. It may be that the problem they point out also invalidates their  similar proof of the zero-one law with respect to finite reflexive frames, corresponding to {\bf T}~\cite[Theorem 5.15 a]{halpern1994}. However, with respect to frame validity for {\bf T}-frames, as far as we know, no counterexample to a zero-one law has yet been published and Le Bars' counterexample cannot easily be adapted to reflexive frames; therefore, the situation remains unsettled for {\bf T}.\footnote{Joe Halpern and Bruce Kapron (personal communication) and Jean-Marie Le Bars (personal communication) confirmed the current non-settledness of the problem for {\bf T}.} 

\subsection{Halpern and Kapron's axiomatization for almost sure frame validities for S4 fails}
Unfortunately, Halpern and Kapron's  proof of the 0-1 law for reflexive, transitive frames and the axiomatization of the almost sure frame validities for reflexive, transitive frames ~\cite[Theorem 5.16]{halpern1994} turn out to be incorrect as well, as follows.\footnote{The author of this paper discovered the counter-example after a colleague had pointed out that the author's earlier attempt at a proof of the 0-1 law for frames of provability logic, inspired by Halpern and Kapron's~\cite{halpern1994} axiomatiation, contained a serious gap.}  Halpern and Kapron introduce the axiom DEP2$'$ and they axiomatize almost-sure frame validities in reflexive transitive frames by {\bf S4}+DEP2$'$~\cite[Theorem 5.16]{halpern1994}, where DEP2$'$ is:$ \neg(p_1 \wedge \Diamond(\neg p_1 \wedge \Diamond (p_1 \wedge \Diamond \neg p_1))).$\\ 

\noindent
The axiom DEP2$'$ precludes $R$-chains $tRuRvRw$ of more than three different states.

\begin{proposition} Suppose $\Phi=\{p_1,p_2\}$. 
Now take the following sentence $\chi$:\up
\[\chi:= (p_1 \wedge \Diamond (\neg p_1 \wedge \Diamond p_1 \wedge \Box (p_1\rightarrow p_2))) \rightarrow \]\vspace{-0.6cm}
\[\Box((\neg p_1 \wedge \Diamond p_1) \rightarrow \Diamond\Box (p_1 \rightarrow  p_2))\vspace{-0.2cm}\]
Then {\bf S4}+DEP2$'\not \vdash\chi$ but $\lim_{n\to\infty} \mu_{n,\Phi}(\chi) = 1$
\end{proposition}

\begin{proof}
It is easy to see that  {\bf S4}+DEP2$'\not \vdash\chi$ by taking the five-point reflexive transitive frame of Figure~\ref{treecounter}, where\up \[M,w_0\models p_1 \wedge \Diamond (\neg p_1 \wedge \Diamond p_1 \wedge \Box(p_1\rightarrow p_2))\up\] but $M, w_3 \not \models (\neg p_1 \wedge \Diamond p_1) \rightarrow \Diamond \Box(p_1 \rightarrow p_2)$, so\up \[M,w_0 \not \models  \Box((\neg p_1 \wedge \Diamond p_1) \rightarrow \Diamond \Box(p_1 \rightarrow p_2)).\] 

Now we sketch a proof that $\chi$ is valid in almost all reflexive Kleitman-Rothschild frames.\footnote{Halpern and Kapron~\cite[Theorem 4.14]{halpern1994} proved that almost surely, every reflexive transitive relation is in fact a partial order, so the Kleitman-Rothschild result also holds for finite frames with reflexive transitive relations.}
So let $M=(W,R,V)$ be an arbitrary large enough (with appropriate extension axioms holding) Kleitman-Rothschild frame $(W, R)$ together with an arbitrary valuation $V$. 
Let $w$ be arbitrary in $W$ and suppose $M,w\models p_1\wedge\Diamond(\neg p_1\wedge\Diamond p_1\wedge \Box(p_1\rightarrow p_2))$. 
Then there is a $w_1\in W$ with $wRw_1$ and $M,w_1\models \neg p_1\wedge\Diamond p_1\wedge \Box(p_1\rightarrow p_2)$. 
We want to show that 
$M,w\models\Box((\neg p_1 \wedge \Diamond p_1) \rightarrow \Diamond \Box(p_1 \rightarrow p_2))$. 
To do this, suppose $w_2$ is arbitrary in $W$ with $wRw_2$ and $M,w_2\models \neg p_1 \wedge \Diamond p_1$. 
The above facts imply that both $w_1$ and $w_2$ are in the middle layer and $w$ is in the bottom layer. 
Then almost surely, there is a $w_3$ in the top layer with $w_1Rw_3$ and $w_2Rw_3$. This confluence follows from Compton’s extension axiom (b)~\cite{compton1988b} (similar to (b) in Proposition~\ref{Axalmost} of the current paper). 
Therefore by $M,w_1\models  \Box(p_1 \rightarrow p_2)$, also $M,w_3\models \Box(p_1 \rightarrow p_2)$, so $M,w_2\models\Diamond\Box(p_1 \rightarrow p_2)$. 
Therefore $M,w\models\Box((\neg p_1 \wedge \Diamond p_1) \rightarrow \Diamond \Box(p_1 \rightarrow p_2))$. 
Now, because $w\in W$ was arbitrary, we have $M\models\chi$. \end{proof}

\begin{figure}
  \begin{minipage}[c]{0.5\textwidth}
  \begin{center}
   \begin{tikzpicture}[x=1.5cm, y=1.4cm]
\node[label=right:{$p_1,\neg p_2$}](mp) at  (1,1) {$w_0$};
\node[label=left:{$\neg p_1, p_2$}] (upq) at  (0,2) {$w_1$};
\node[label=right:{$\neg p_1,\neg p_2$}](uq) at  (2,2) {$w_3$};
\node[label=left:{$ p_1, p_2$}]  (tl) at  (0,3) {$w_2$};
\node[label=right:{$p_1,\neg p_2$}](tr) at  (2,3) {$w_4$};
\draw[->] (mp) -- (upq);
\draw[->] (mp) -- (uq);
\draw[->] (upq) -- (tl);
\draw[->] (uq) -- (tr);
\end{tikzpicture}
\end{center}
  \end{minipage}\hfill
  \begin{minipage}[c]{0.49\textwidth}
    \caption{Counter-model showing that 
    the formula $\chi:=$
    $(p_1 \wedge \Diamond (\neg p_1 \wedge \Diamond p_1 \wedge \Box (p_1\rightarrow p_2))) \rightarrow 
\Box((\neg p_1 \wedge \Diamond p_1) \rightarrow \Diamond\Box (p_1 \rightarrow  p_2))$
 does not hold in $w_0$ of this three-layer model. The relation in the model is the reflexive transitive closure of the one represented by the arrows. }
    \label{treecounter}
  \end{minipage}
  \end{figure}
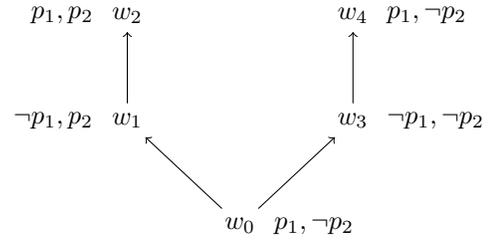

\noindent
Therefore, the axiom system given in ~\cite[Theorem 5.16]{halpern1994} is {\em not} complete with respect to almost-sure frame validities for finite reflexive transitive orders.
Fortunately, it seems possible to mend the situation and still obtain an axiom system that is sound and complete with respect to almost sure $\mathcal{S}$4 frame validity,  by adding extra axioms characterizing the umbrella- and diamond properties that we will also use for the provability logic {\bf GL} in Section~\ref{Frames}; the $\mathcal{S}$4 version is future work.

\subsection{Almost sure model validity does not coincide with almost sure frame validity}
Interestingly, whereas for full classes of frames, `validity in {\em all} finite models' coincides with `validity in {\em all} finite frames' of the class, this is not the case for `almost sure validity'. In particular, for both the class of reflexive transitive frames ($\mathcal{S}$4) and the class of reflexive transitive symmetric frames ($\mathcal{S}$5), there are many more formulas that are `valid in {\em almost all} finite models' than `valid in {\em almost all} finite frames' of the appropriate kinds. 
Our work has been greatly inspired by Halpern and Kapron's paper~\cite{halpern1994} and we also use some of the previous results that they applied, notably the above-mentioned combinatorial result by Kleitman and Rothschild about finite partial orders.\\

 
 \noindent
 The rest of this paper is structured as follows. In Section~\ref{Provability-intro}, we give a brief reminder of the axiom system and semantics of provability logic. In the central Sections~\ref{GL-models},~\ref{Random} and~\ref{Frames}, we show why provability logic obeys zero-one laws both with respect to its models and with respect to its  frames. We provide two axiom systems characterizing the formulas that are almost always valid in the relevant models, respectively almost always  valid in the relevant frames. When discussing almost sure frame validity, we will investigate both the almost sure validity in finite irreflexive transitive frames and validity in the countable random Kleitman-Rothschild frame, and show that there is transfer between them. Section~\ref{Complexity}  provides a sketch of the complexity of the decidability problems of almost sure model and almost sure frame validity for provability logic.
 Finally, Section~\ref{Discussion} presents a conclusion and some questions for future work. 
 
 The result on models in Section~\ref{GL-models} was proved 26 years ago, and presented in \cite{Verbrugge1995,verbrugge2018}, but the proofs have not been published before in an archival venue. 
 The results about almost sure frame validities for {\bf GL} are new, as well as the counter-example against the axiomatization by Halpern and Kapron of almost sure $\mathcal{S}$4 frame validities.\footnote{Due to the length restriction, this paper includes proof sketches of the main results. Full proofs are to be included in an extended version for a journal.}
 
 \section{Provability logic}
 \label{Provability-intro}
\label{GL}

In this section, a brief reminder is provided about the protagonist of this paper: the provability logic {\bf GL},
named after G\"{o}del and L\"{o}b. As axioms, it contains all axiom schemes from $\mathbf{K}$ and the extra scheme GL. Here follows the full set of axiom schemes of $\mathbf{GL}$:\vspace{-0.14cm}
\begin{align}
\tag{A1}
&\text{All (instances of) propositional tautologies} \label{eq:A1}\\
\tag{A2}
&\square(\varphi\rightarrow \psi) \rightarrow (\square\varphi \rightarrow \square\psi) \label{eq:A2}\\
\tag{GL}
&\square (\square \varphi \rightarrow \varphi)\rightarrow \square \varphi \vspace{-0.16cm}
\label{eq:GL}
\end{align}
 The rules of inference are modus ponens and necessitation:\vspace{-0.02cm}
 \begin{quote}
 if $\mathbf{GL}\vdash \varphi\rightarrow \psi$ and $\mathbf{GL}\vdash \varphi$, then $\mathbf{GL}\vdash\varphi$.\\
 if $\mathbf{GL}\vdash \varphi$, then $\mathbf{GL}\vdash \square\varphi$.
 \end{quote}
  Note that the transitivity axiom $\square \varphi \rightarrow \square \square \varphi$ follows from $\mathbf{GL}$, which was first proved by De Jongh and Sambin~\cite{boolos1993,Verbrugge2017}, but that the reflexivity axiom $\square \varphi \rightarrow \varphi$ does not follow.
Indeed, Segerberg proved in 1971 that provability logic is sound and complete with respect to all transitive, {\em converse well-founded} frames (i.e., for each non-empty set $X$, there is an R-greatest element of $X$; or equivalently: there is no infinitely ascending sequence $x_1Rx_2Rx_3Rx_4, \ldots$). Segerberg also proved completeness with respect to all finite, transitive, irreflexive frames~\cite{segerberg1971}. 
The latter soundness and completeness result will be relevant for our purposes. For more information on provability logic, see, for example,~\cite{smorynski1985,boolos1993,Verbrugge2017}.

\label{Zero-one-provability}

In the next three sections, we provide axiomatizations, first for almost sure model validity and then for almost sure frame validity, with respect to the relevant finite frames corresponding to {\bf GL}, namely the irreflexive transitive ones.

For the proofs of the zero-one laws for almost sure model and frame validity, we will need completeness proofs of the relevant axiomatic theories -- let us refer to such a theory by $\mathbf{S}$ for the moment -- with respect to almost sure model validity and with respect to almost sure frame validity. Here we will use Lindenbaum's lemma and maximal $\mathbf{S}$-consistent sets of formulas. For such sets, the following useful properties hold, as usual~\cite{segerberg1971,chellas1980}: 

\begin{proposition}
\label{maxconsistent}
Let $\Theta$ be a maximal $\mathbf{S}$-consistent set of formulas in $L(\Phi)$. Then for each pair of formulas $\varphi, \psi \in L(\Phi)$:
                        \begin{enumerate}
                                \item  $\varphi \in \Theta$ iff $\neg \, \varphi \not\in \Theta$;
                                \item $(\varphi \wedge \psi) \in \Theta \Leftrightarrow \varphi \in \Theta$ and $\psi \in \Theta$;
                                \item if $\varphi \in \Theta$ and $(\varphi \rightarrow \psi) 
                                \in \Theta$ then $\psi \in \Theta$;
                                \item if $\Theta \vdash_{\mathbf{S}} \; \varphi$ then $\varphi \in \Theta$.
                        \end{enumerate}
\end{proposition}

\section{Validity in almost all finite irreflexive transitive models}
\label{GL-models}

The axiom system $\mathbf{AX^{\Phi,M}_{GL}}$ has the same axioms and rules as {\bf GL} (see Section~\ref{GL}) plus the following axioms:
\begin{align}
\tag{T3} 
&\Box\Box\Box\bot \label{Threelayers}\\
\tag{C1}
& \Diamond \top \rightarrow \Diamond A \label{Carnap1}\\
\tag{C2}
& \Diamond \Diamond \top \rightarrow \Diamond(B \wedge \Diamond C) \label{Carnap2}
\end{align}

\noindent

\noindent
In the axiom schemes C1 and C2, the formulas $A$, $B$ and $C$ all stand for consistent conjunctions of literals over $\Phi$. 

These axiom schemes have been inspired by Carnap's consistency axiom: $\Diamond \varphi$ for any $\varphi$ that is a consistent propositional formula~\cite{Carnap1947}, which has been used by Halpern and Kapron~\cite{halpern1994} for axiomatizing almost sure model validities for $\mathcal{K}$-models. 

Note that $\mathbf{AX^{\Phi,M}_{GL}}$  is not a normal modal logic, because one cannot substitute just any formula for $A, B, C$; for example, substituting $p_1\wedge \neg p_1$ for $A$ in C1 would make that formula equivalent to $\neg\Diamond\top$, which is clearly undesired. However, even though $\mathbf{AX^{\Phi,M}_{GL}}$ is not closed under uniform substitution, it is still a propositional theory, in the sense that it is closed under modus ponens.

\begin{example}
For $\Phi=\{p_1, p_2\}$, the axiom scheme C1 boils down to the following four axioms:
\begin{align}
& \Diamond \top \rightarrow \Diamond (p_1 \wedge p_2)\\
& \Diamond \top \rightarrow \Diamond (p_1 \wedge \neg p_2)\\
&\Diamond \top \rightarrow \Diamond (\neg p_1 \wedge p_2)\\
&\Diamond \top \rightarrow \Diamond (\neg p_1 \wedge \neg p_2)
 \end{align}
The axiom scheme C2 covers 16 axioms, corresponding to the $2^4$ possible choices of positive or negative literals, as captured by the following scheme, where ``$[\neg]$'' is shorthand for a negation being present or absent at the current location:
\[\Diamond\Diamond \top \rightarrow \Diamond ([\neg]p_1 \wedge [\neg]p_2  \wedge \Diamond([\neg]p_1 \wedge [\neg]p_2))\]
\end{example}

\noindent
The following definition of the canonical asymptotic model over a finite set of propositional atoms $\Phi$ is based on the set of propositional valuations on $\Phi$, namely, the functions $v$ from the set of propositional atoms $\Phi$ to the set of truth values $\{0,1\}$. As worlds, we introduce for each such valuation $v$ three new distinct objects, for mnemonic reasons called $u_v$ (Upper), $m_v$ (Middle), and $b_v$ (Bottom); see Figure 2.

\begin{definition} 
\label{Canonical-GL}
Define $\mathrm{M}^{\Phi}_{GL}= (W, R, V)$, the {\em canonical asymptotic model} over $\Phi$, with $W, R, V$ as follows:\\
$W= \{b_v \mid v \mbox{ a propositional valuation on } \Phi \} \cup \\
\mbox{   \hspace{0.55cm}    } \{m_v \mid v \mbox{ a propositional valuation on } \Phi \} \cup \\
\mbox{   \hspace{0.55cm}    }  \{u_v \mid v \mbox{ a propositional valuation on } \Phi \}$\\
 $R=\{\langle b_v, m_{v'}\rangle \mid v, v' \mbox{ propositional valuations on } \Phi \} \cup \\
\mbox{   \hspace{0.55cm}    }\{\langle m_v, u_{v'}\rangle \mid v, v' \mbox{ propositional valuations on } \Phi \} \cup \\
\mbox{   \hspace{0.55cm}    }  \{\langle b_v, u_{v'}\rangle \mid v, v' \mbox{ propositional valuations on } \Phi \}$; \\ 
and
for all $p_i\in\Phi$ and all propositional valuations $v$ on $\Phi$, the modal valuation $V$ is defined by:\\
$V_{b_v} (p_i) = V_{m_v} (p_i) = V_{u_v} (p_i) = v(p_i)$.\footnote{ If $\Phi$ were enumerably infinite, the definition could be adapted so that precisely those propositional valuations are used that make only finitely many propositional atoms true, see also~\cite{halpern1994}.}
\end{definition}
 
 
 \begin{figure*}
 \label{canonical16}
 \begin{tikzpicture}[x=4.9cm, y=2cm]
\node[label=below:{$p_1,p_2$}] (bpq) at  (0,0) {$b_{v_1}$};
\node[label=below:{$p_1, \neg p_2$}](bp) at  (1,0) {$b_{v_2}$};
\node[label=below:{$\neg p_1, p_2$}] (bq) at  (2,0) {$b_{v_3}$};
\node[label=below:{$\neg p_1, \neg p_2$}] (b) at  (3,0) {$b_{v_4}$};
\node[label=left:{$p_1,p_2$}] (mpq) at  (0,1) {$m_{v_1}$};
\node[label=left:{$p_1, \neg p_2$}] (mp) at  (1,1) {$m_{v_2}$};
\node[label=right:{$\neg p_1, p_2$}] (mq) at  (2,1) {$m_{v_3}$};
\node[label=right:{$\neg p_1, \neg p_2$}] (m) at  (3,1) {$m_{v_4}$};
\node[label=above:{$p_1,p_2$}]  (upq) at  (0,2) {$u_{v_1}$};
\node[label=above:{$p_1, \neg p_2$}] (up) at  (1,2) {$u_{v_2}$};
\node[label=above:{$\neg p _1, p_2$}] (uq) at  (2,2) {$u_{v_3}$};
\node[label=above:{$\neg p _1, \neg p_2$}] (u) at  (3,2) {$u_{v_4}$};
\draw[->] (bpq) -- (mpq);
\draw[->] (bpq) -- (mp);
\draw[->] (bpq) -- (mq);
\draw[->] (bpq) -- (m);
\draw[->] (bp) -- (mpq);
\draw[->] (bp) -- (mp);
\draw[->] (bp) -- (mq);
\draw[->] (bp) -- (m);

\draw[->] (bq) -- (mpq);
\draw[->] (bq) -- (mp);
\draw[->] (bq) -- (mq);
\draw[->] (bq) -- (m);

\draw[->] (b) -- (mpq);
\draw[->] (b) -- (mp);
\draw[->] (b) -- (mq);
\draw[->] (b) -- (m);

\draw[->] (mpq) -- (upq);
\draw[->] (mpq) -- (up);
\draw[->] (mpq) -- (uq);
\draw[->] (mpq) -- (u);
\draw[->] (mp) -- (upq);
\draw[->] (mp) -- (up);
\draw[->] (mp) -- (uq);
\draw[->] (mp) -- (u);

\draw[->] (mq) -- (upq);
\draw[->] (mq) -- (up);
\draw[->] (mq) -- (uq);
\draw[->] (mq) -- (u);

\draw[->] (m) -- (upq);
\draw[->] (m) -- (up);
\draw[->] (m) -- (uq);
\draw[->] (m) -- (u);
\end{tikzpicture}
\caption{The canonical asymptotic model $\mathrm{M}^{\Phi}_{GL}= (W, R, V)$ for $\Phi=\{p_1, p_2\}$, defined in Definition~\ref{Canonical-GL}. The accessibility relation is the transitive closure of the relation given by the arrows drawn in the picture. The four relevant valuations are $v_1, v_2, v_3, v_4$, given by $v_1(p_1)=v_1(p_2)=1$; $v_2(p_1)=1, v_2(p_2)=0$; $v_3(p_1)=0, v_3(p_2)=1$; $v_4(p_1)=v_4(p_2)=0$.} 
\end{figure*}
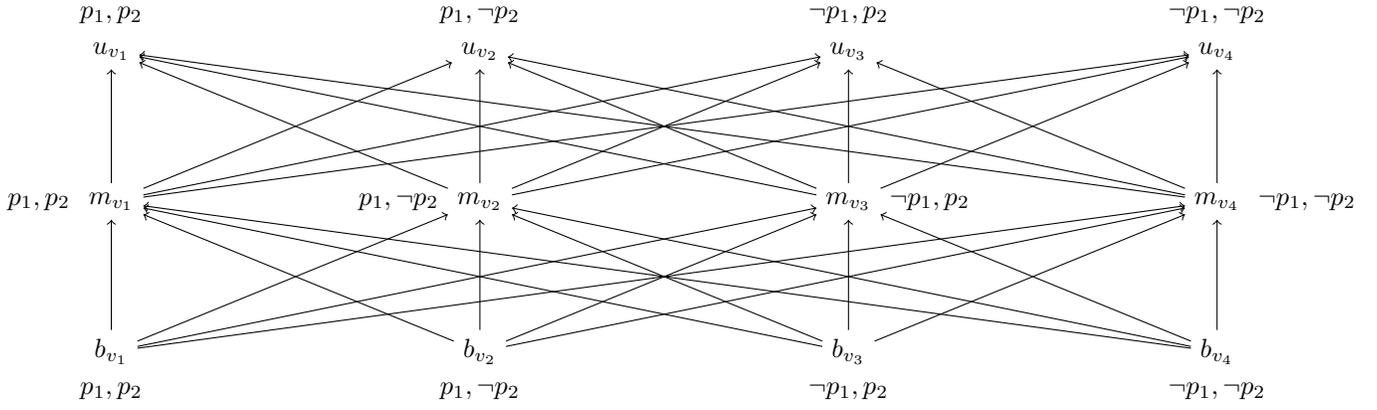

\noindent
For the proof of the zero-one law for model validity, we will need a completeness proof of $\mathbf{AX^{\Phi,M}_{GL}}$ with respect to almost sure model validity, including use of Lindenbaum's lemma and Proposition~\ref{maxconsistent}, applied to 
$\mathbf{AX^{\Phi,M}_{GL}}$.
The zero-one law for model validity follows directly from the following theorem:
\begin{theorem}
For every formula $\varphi \in L(\Phi)$, the following are equivalent:
\begin{enumerate}
\item $\mathrm{M}^{\Phi}_{GL}\models \varphi$;
\item $\mathbf{AX^{\Phi,M}_{GL}} \vdash \varphi$;
\item $\lim_{n\to\infty} \nu_{n,\Phi}(\varphi) = 1$;
\item $\lim_{n\to\infty} \nu_{n,\Phi}(\varphi) \not = 0$.
\end{enumerate}
\end{theorem}

\begin{proof} We show a circle of implications. Let $\varphi\in L(\Phi)$.\\

\noindent
{\bf 1 $\Rightarrow$ 2}\\
 By contraposition. Suppose that $\mathbf{AX^{\Phi,M}_{GL}} \not \vdash \varphi$, then $\neg \varphi$ is $\mathbf{AX^{\Phi,M}_{GL}}$-consistent. By Lindenbaum's lemma, we can extend $\{\neg \varphi\}$ to a maximal $\mathbf{AX^{\Phi,M}_{GL}}$-consistent set $\Delta$ over $\Phi$. We use a standard canonical model construction; here, we illustrate how that works for the finite set $\Phi = \{p_1, p_2\}$, but the method works for any finite $\Phi=\{p_1, \ldots, p_k\}$.\footnote{For adapting to the enumerably infinite case, see~\cite[Theorem 4.15]{halpern1994}.} It will turn out that the model we define is isomorphic to the model of Definition~\ref{Canonical-GL}. Let us define the model  $\mathrm{MC}^{\Phi}_{GL}= (W', R', V')$:

\begin{itemize}
\item $W'= \{ w_{\Gamma} \mid \Gamma \mbox { is maximal } \mathbf{AX^{\Phi,M}_{GL}}\mbox{-consistent,}$\\
$\mbox{ } \hspace{2cm}\mbox{    based on } \Phi \}$.
\item $R' = \{ \langle w_{\Gamma_1}, w_{\Gamma_2}\rangle \mid w_{\Gamma_1}, w_{\Gamma_2}\in W' \mbox{ and } $\\$ \mbox{ }\hspace{2cm}\mbox{ for all } \Box\psi\in \Gamma_1, \mbox{ it holds that } \psi \in \Gamma_2 \}$
\item For each $w_{\Gamma} \in W': V'_{w_{\Gamma}}(p)= 1 \mbox{ iff } p\in \Gamma$
\end{itemize}

\noindent
Because the worlds of this model correspond to the maximal $\mathbf{AX^{\Phi,M}_{GL}}$-consistent sets, 
all worlds $w_{\Gamma}\in W'$  can be distinguished into three kinds, exhaustively and without overlap:

\begin{description}
\item[U] $\Box\bot \in \Gamma$; there are exactly four maximal consistent sets $\Gamma$ of this form, determined by which of the four conjunctions of relevant literals $[\neg] p_1 \wedge [\neg] p_2$ is an element. These comprise the upper level U of the model.
\item[M] $\neg \Box\bot \in \Gamma$ and $\Box\Box\bot \in \Gamma$;  there are exactly four maximal consistent sets $\Gamma$ of this form, determined by which of the four conjunctions of relevant literals $[\neg] p_1 \wedge [\neg] p_2$ is an element. By axiom C1 and Proposition~\ref{maxconsistent}, all these four maximal consistent sets contain the four formulas of the form $\Diamond([\neg]p_1 \wedge [\neg] p_2)$; by definition of $R'$ and using the fact that $\Box\Box\bot \in \Gamma$, this means that all the four worlds in this middle level M will have access to all the four worlds in the upper level U.
\item[B] $\neg \Box\bot \in \Gamma$ and $\neg\Box\Box\bot \in \Gamma$ and $\Box\Box\Box\bot \in \Gamma$;  there are exactly four maximal consistent sets $\Gamma$ of this form, determined by which of the four conjunctions of relevant literals $[\neg] p_1 \wedge [\neg] p_2$ is an element.  Because $\Diamond\Diamond\top \in \Gamma$, by axiom C2 and Proposition~\ref{maxconsistent}, all these four maximal consistent sets contain the 16 formulas $\Diamond([\neg]p_1 \wedge [\neg] p_2 \wedge \Diamond([\neg]p_1 \wedge [\neg] p_2))$. By the definition of $R'$, this means that all four worlds in this bottom level B will have direct access to all the four worlds in middle level M as well as access in two steps to all four worlds in upper level U.\\
\end{description}

\noindent
Note that $R'$ is transitive because $\mathbf{AX^{\Phi,M}_{GL}}$ extends {\bf GL}, so for all maximal consistent sets $\Gamma$ and all formulas $\psi \in L(\Phi)$, we have that $\Box\psi \rightarrow \Box\Box\psi\in\Gamma$.
Also $R'$ is irreflexive. Because each world contains either $\Box\bot$ and $\neg\bot$ (for U), or  $\Box\Box\bot $ and $\neg \Box\bot$ (for M), or $\Box\Box\Box\bot$ and $\neg\Box\Box\bot$ (for B), by definition of $R'$, none of the worlds has a relation to itself.\\

\noindent
The next step is to prove by induction that a {\em truth lemma} holds: For all $\psi$ in the language $L(\Phi)$ and for all maximal $\mathbf{AX^{\Phi,M}_{GL}}$-consistent sets $\Gamma$, the following holds:

\begin{quote} 
$\mathrm{MC}^{\Phi}_{GL}, w_{\Gamma} \models \psi$ iff $\psi \in \Gamma$.\\
\end{quote}
\noindent
 For atoms, this follows by the definition of $V'$. The steps for the propositional connectives are as usual, using the properties of maximal consistent sets (see Proposition~\ref{maxconsistent}).\\
 
 \noindent
 For the $\Box$-step, let $\Gamma$ be a maximal  $\mathbf{AX^{\Phi,M}_{GL}}$-consistent set and let  us suppose as induction hypothesis that for some arbitrary formula $\chi$, for all maximal $\mathbf{AX^{\Phi,M}_{GL}}$-consistent sets $\Pi$, $\mathrm{MC}^{\Phi}_{GL}, w_{\Pi} \models \chi$ iff $\chi \in \Pi$. We want to show that $\mathrm{MC}^{\Phi}_{GL}, w_{\Gamma} \models \Box\chi$ iff $\Box\chi \in \Gamma$. 

For the direction from right to left, suppose that $\Box\chi\in \Gamma$, then by definition of $R'$, for all $\Pi$ with $w_{\Gamma} R' w_{\Pi}$, we have $\chi\in\Pi$, so by induction hypothesis, $\mathrm{MC}^{\Phi}_{GL}, w_{\Pi}\models \chi$. Therefore, by the truth definition, $\mathrm{MC}^{\Phi}_{GL}, w_{\Gamma}\models \Box\chi$.

For the direction from left to right, let us use contraposition and suppose that $\Box\chi\not\in \Gamma$. 
Now we will show that the set $\{\xi \mid \Box\xi\in \Gamma\}\cup \{\neg \chi\}$ is $\mathbf{AX^{\Phi,M}_{GL}}$-consistent. For otherwise, there would be some $\xi_1,\ldots,\xi_n$ for which $\Box \xi_1, \ldots, \Box \xi_n \in \Gamma$ such that $\xi_1,\ldots,\xi_n \vdash_{\mathbf{AX^{\Phi,M}_{GL}}} \; \chi$, so by necessitation, A2, and propositional logic, $\Box\xi_1,\ldots,\Box\xi_n \vdash_{\mathbf{AX^{\Phi,M}_{GL}}} \; \Box\chi$, therefore by maximal consistency of $\Gamma$ and Proposition~\ref{maxconsistent}(iv), also $\Box\chi\in\Gamma$, contradicting our assumption. Therefore, by Lindenbaum's lemma there is a maximal consistent set $\Pi \supseteq \{\xi \mid \Box\xi\in \Gamma\}\cup \{\neg \chi\}$. It is clear by definition of $R'$ that $w_{\Gamma} R' w_{\Pi}$, and by induction hypothesis, $\mathrm{MC}^{\Phi}_{GL}, w_{\Pi} \models \neg\chi$, i.e., $\mathrm{MC}^{\Phi}_{GL}, w_{\Pi} \not\models\chi$, so by the truth definition, $\mathrm{MC}^{\Phi}_{GL}, w_{\Gamma} \not\models\Box\chi$. This finishes the inductive proof of the truth lemma.\\

\noindent
Finally, from the truth lemma and the fact stated at the beginning of the proof of 1 $\Rightarrow$ 2 that $\neg\varphi \in \Delta$, we have that $\mathrm{MC}^{\Phi}_{GL}, w_{\Delta} \not \models \varphi$, so we have found our counter-model.\\

\noindent
With its three layers (Upper, Middle, and  Bottom) of four worlds each, corresponding to each consistent conjunction of literals, and with each world corresponding to maximal consistent sets containing axioms C1 and C2 and therefore being related to precisely all those worlds in the layers above, the model $\mathrm{MC}^{\Phi}_{GL}$ that we construct in the completeness proof above is isomorphic to the canonical asymptotic model $\mathrm{M}^{\Phi}_{GL}$ of Definition~\ref{Canonical-GL}; for $\Phi=\{p_1, p_2\}$, see Figure 2.\\

\noindent
{\bf 2 $\Rightarrow$ 3}\\ Suppose that $\mathbf{AX^{\Phi,M}_{GL}}\vdash \varphi$. We will show that the axioms of $\mathbf{AX^{\Phi,M}_{GL}}$ hold in almost all irreflexive transitive Kleitman-Rothschild models of depth 3 (see Subsection~\ref{KR}).  First, it is immediate that {\bf GL} is sound with respect to {\em all} finite irreflexive transitive converse well-founded models, that axiom $\Box\Box\Box\bot$ is sound with respect to those  of depth 3, and that almost sure model validity is closed under MP and Necessitation. It remains to show the almost sure model validity of axiom schemes C1 and C2 over finite irreflexive models of the Kleitman-Rothschild variety.\\
 
 \noindent
We will now sketch a proof that the `Carnap-like' axiom C1, namely $\Diamond \top \rightarrow \Diamond A$ where $A$ is a consistent conjunction of literals over $\Phi$, is valid in almost all irreflexive transitive models $(W,R,V)$ of depth 3 of the Kleitman-Rothschild variety with an arbitrary $V$. 
Let us suppose that $\Phi=\{p_1,\ldots,p_k\}$, so there are $2^k$ possible valuations. Let us consider a state $s$ in such a model of $n$ elements where $\Diamond\top$ holds; then, being a Kleitman-Rothschild model, $s$ has as direct successors approximately half of the states in the directly higher layer, which contains asymptotically at least $\frac{1}{4}$ of the model's states. So $s$ has asymptotically at least $\frac{1}{8}\cdot n$ direct successors. The probability that a given state $t$ is a direct successor of $s$ with the right valuation to make $A$ true is therefore at least $\frac{1}{8}\cdot \frac{1}{2^k} = \frac{1}{2^{k+3}}$. Thus, the probability that $s$ does not have {\em any} direct successors in which $A$ holds is at most $(1- \frac{1}{2^{k+3}})^n$. Therefore, the probability that there is at least one $s$ in a Kleitman-Rothschild model not having any direct successors satisfying $A$ is at most $n\cdot (1- \frac{1}{2^{k+3}})^n$. By a standard calculus result, $\lim_{n\to\infty} n\cdot (1- \frac{1}{2^{k+3}})^n = 0$, so C1 is valid in almost all Kleitman-Rothschild models,  i.e., $\lim_{n\to\infty} \nu_{n,\Phi}(\Diamond \top \rightarrow \Diamond A)=1$.\\

\noindent
Similarly, we sketch a proof that axiom C2, namely $\Diamond\Diamond \top \rightarrow \Diamond (B \wedge \Diamond C)$ where $B, C$ are consistent conjunctions of literals over $\Phi$, is valid in almost all irreflexive transitive Kleitman-Rothschild models $(W,R,V)$ of depth 3 with an arbitrary $V$. Let $\Phi=\{p_1,\ldots,p_k\}$. Again, let us consider a state $s$ in such a model of $n$ elements where $\Diamond\Diamond\top$ holds, then $s$ is in the bottom of the three layers; therefore, the model being of Kleitman-Rothschild type, $s$ has as direct successors approximately half of the states in the middle layer, which contains asymptotically at least $\frac{1}{2}$ of the model's states. So $s$ has asymptotically at least $\frac{1}{4}\cdot n$ direct successors. 

The probability that a given state $t$ is a direct successor of $s$ with the right valuation to make $B$ true is therefore at least $\frac{1}{4}\cdot \frac{1}{2^k} = \frac{1}{2^{k+2}}$. Similarly, given such a $t$, the probability that a given state $t'$ in the top layer is a direct successor of $t$ in which $C$ holds is asymptotically at least $\frac{1}{2^{k+2}}\cdot \frac{1}{2^{k+3}}= \frac{1}{2^{2k+5}}$ Therefore, the probability that for the given $s$ there are {\em no} $t, t'$ with $sRtRt'$ with $B$ true at $t$ and $C$ true at $t'$ is at most $(1-\frac{1}{2^{2k+5}})^n$. Summing up, the probability that there is at least one $s$ in a Kleitman-Rothschild model not having any pair of successors $sRtRt'$ with $B$ true at $t$ and $C$ true at $t'$ is at most $n\cdot (1-\frac{1}{2^{2k+5}})^n$. Again, $\lim_{n\to\infty} n\cdot (1-\frac{1}{2^{2k+5}})^n=0$, so C2 holds in almost all Kleitman-Rothschild models, i.e. $\lim_{n\to\infty} \nu_{n,\Phi}(\Diamond\Diamond \top \rightarrow \Diamond (B \wedge \Diamond C)) = 1$.\\

\noindent
{\bf 3 $\Rightarrow$ 4}\\
 Obvious, because $0 \not = 1$.\\
 
 \noindent
{\bf 4 $\Rightarrow$ 1}\\ By contraposition. Suppose as before that $\Phi=\{p_1,\ldots,p_k\}$. Now suppose that  the canonical asymptotic model $\mathrm{M}^{\Phi}_{GL}\not\models \varphi$ for some $\varphi \in L(\Phi)$, say, $\mathrm{M}^{\Phi}_{GL}, s\not\models \varphi$, for some $s  \in W$. We claim that almost surely for a sufficiently large finite Kleitman-Rothschild type model $M' = (W', R', V')$ of three layers with $V'$ random, there is a bisimulation relation $Z$ from $\mathrm{M}^{\Phi}_{GL}$ to $M'$. 
We now sketch how to define the 
bisimulation $Z$. 
 
 Asymptotically, we will be able to find in $M'$ a world $s'$ that is situated at the same layer (top, middle or bottom) as the layer where $s$ is in $\mathrm{M}^{\Phi}_{GL}$ and that has the same valuation for all atoms $p_1, \ldots, p_k$. First, fix an enumeration of all $2^k$ possible valuations: $v_1, \ldots, v_{2_k}$. For each $b_{v_i}$ with valuation $v_i$ (where $i\in \{1,\ldots,2^k\}$) in the bottom layer of $\mathrm{M}^{\Phi}_{GL}$, there will almost surely be a $b'_i\in W'$ that has the same valuation $v_i$, as well as direct successors $m'_{i,1},\ldots,m'_{i,2^k}$ corresponding to valuations $v_1, \ldots, v_{2^k}$ respectively, where each $m'_{i,j}$ in turn has $2^k$ direct successors $u'_{i,j,1}, \ldots, u'_{i,j,2^k}$ corresponding to valuations $v_1, \ldots, v_{2^k}$ respectively.

Take relation $Z$ given by: for all $i, j, l \in \{1,\ldots,2^k\}$, $b_{v_i} Z b'_i$ and $m_{v_j} Z m'_{i,j}$ and $u_{v_l} Z u'_{i,j,l}$. This $Z$ satisfies the three conditions for bisimulations for all $w\in W, w'\in W'$: If $wZw'$, then $w$ and $w'$ have the same valuation; the `forth' condition that $wZw'$ and $wRv$ together imply that there is a $v'\in W'$ such that $w'R'v'$ and $vZv'$; and the `back' condition that $wZw'$ and $w'R'v'$ together imply that there is a $v\in W$ such that $wRv$ and $vZv'$. 

Now that the bisimulation $Z$ is given, suppose that $sZs'$ for some $s'\in W'$. By the bisimulation theorem~\cite{Benthem1983}, we have that for all $\psi\in L(\Phi)$,  $\mathrm{M}^{\Phi}_{GL}, s\models\psi \Leftrightarrow M',s'\models\psi$, in particular, $M'\not\models\varphi$. Conclusion: $\lim_{n\to\infty} \nu_{n,\Phi}(\varphi)=0$. 
\\

\noindent
We can now conclude that all of 1, 2, 3, 4 are equivalent. Therefore, each modal formula in $L(\Phi)$ is either almost surely valid or almost surely invalid over finite models in $\mathcal{GL}$.
\end{proof}

\noindent
This concludes our investigation of validity in almost all models. For almost sure frame validity, it turns out that there is transfer between validity in the countable irreflexive Kleitman Rothschild frame and almost sure frame validity.

\section{The countable random irreflexive Kleitman-Rothschild frame}
\label{Random}

Differently than for the system {\bf K}~\cite{goranko2003}, it turns out that in logics for transitive partial (strict) orders such as {\bf GL}, we can prove transfer between validity of a sentence in almost all relevant finite frames and validity of the sentence in one specific frame, namely the countably infinite random irreflexive Kleitman Rothschild frame. Let us start by introducing this frame step by step.




\noindent
The following definition specifies a first-order theory in the language of strict (irreflexive asymmetric) partial orders. We have adapted it from Compton's~\cite{compton1988} set of extension axioms $T_\mathrm{as}$ (where the subscript ``{\it as}'' stands for ``
almost sure'') for reflexive partial orders of the Kleitman-Rothschild form, which were in turn inspired by Gaifman's and Fagin's extension axioms for almost all first-order models with a binary relation~\cite{gaifman1964,fagin1976}.

\begin{definition}[Extension axioms]
\label{Extension}
The theory $T_\mathrm{as\mbox{-}irr}$\footnote{Here, the subscript {\it as-irr} stands for ``almost sure - irreflexive''.} includes the axioms for strict partial orders, namely, $\forall x \neg(x < x)$ and $\forall x, y, z((x < y \wedge y < z) \rightarrow x < z)$. In addition, it includes the following:
\begin{align}
\tag{Depth-at-least-3} 
& \exists x_0, x_1, x_2,(\bigwedge_{i\leq 1} x_i < x_{i+1})\label{Threelayers-02}
\end{align}
\begin{align}
\tag{Depth-at-most-3} 
& \neg \exists x_0, x_1, x_2, x_3 (\bigwedge_{i\leq 2} x_i < x_{i+1})\label{Threelayers-2}
\end{align}
Every strict partial order satisfying Depth-at-least-3 and Depth-at-most-3 can be partitioned into the three levels $L_1$ (Bottom), $L_2$ (Middle), and $L_3$ (Upper) as in Subsection~\ref{KR} and these levels are first-order definable. Let us describe the extension axioms. 

\noindent
For every $j, k, l \geq 0$ there is an extension axiom saying that for all distinct $x_0,\ldots,x_{k-1}$ and $y_0,\ldots,y_{j-1}$ in $L_2$ and all distinct $z_0,\ldots,z_{l-1}$ in $L_1$, there is an element $z$ in $L_1$ not equal to $z_0,\ldots,z_{l-1}$ such that:\vspace{-0.05cm}
\begin{align}
\tag{a}
& \bigwedge_{i < k}   z < x_i \; \wedge  \bigwedge_{i <j}   \neg(z < y_i) \label{Extension-a}
\end{align}
For every $j, k, l \geq 0$ there is an axiom saying that for all distinct $x_0,\ldots,x_{k-1}$ and $y_0,\ldots,y_{j-1}$ in $L_2$ and all distinct $z_0,\ldots,z_{l-1}$ in $L_3$, there is an element $z$ in $L_3$ not equal to $z_0,\ldots,z_{l-1}$ such that:
\begin{align}
\tag{b}
& \bigwedge_{i < k}   x_i < z \; \wedge  \bigwedge_{i <j}   \neg(y_i < z) \label{Extension-b}
\end{align}
For every $j, j', k, k', l \geq 0$ there is an axiom saying that for all distinct $x_0,\ldots,x_{k-1}$ and $y_0,\ldots,y_{j-1}$ in $L_1$ and all distinct $x'_0,\ldots,x'_{k'-1}$ and $y'_0,\ldots,y'_{j'-1}$ in $L_3$, and all distinct $z_0,\ldots,z_{l-1}$ in $L_2$, there is an element $z$ in $L_2$ not equal to $z_0,\ldots,z_{l-1}$ such that:
\begin{align}
\tag{c}
& \bigwedge_{i < k}   x_i < z \; \wedge  \bigwedge_{i <j}   \neg(y_i < z)  \; \wedge \bigwedge_{i < k'}   z < x'_i \; \wedge  \bigwedge_{i <j'}   \neg(z < y'_i) \vspace{-0.2cm} \label{Extension-c}
\end{align}\vspace{-0.3cm}
\end{definition}
\noindent
\begin{proposition} 
\label{Categorical}
$T_\mathrm{as\mbox{-}irr}$ is $\aleph_0$-categorical and therefore also complete, because it has no finite models.
\end{proposition}

\noindent
{\bf Proof sketch} 
Straightforward adaptation from Compton's reflexive to our irreflexive orders of his proof that his $T_\mathrm{as}$  is $\aleph_0$-categorical and therefore also complete~\cite[Theorem 3.1]{compton1988b}. 

Let us call the unique countably infinite model of $T_\mathrm{as\mbox{-}irr}$ by the name $\mathcal{F}_{KR}$: the countable random irreflexive Kleitman-Rothschild frame.

\begin{proposition}
\label{Axalmost}
Each of the sentences in $T_\mathrm{as\mbox{-}irr}$ has labeled asymptotic probability 1 in the class of finite strict (irreflexive) partial orders.
\end{proposition}

\noindent
{\bf Proof sketch} Straightforward adaptation to our irreflexive orders of Compton's proof that his $T_\mathrm{as}$ has labeled asymptotic probability 1 in reflexive partial orders~\cite[Theorem 3.2]{compton1988b}.\up\\

\noindent
Notice that by extension axiom (c) and transitivity, in almost all Kleitman-Rothschild models, all points in the bottom layer $L_1$ are connected to all points in the top layer $L_3$.\\

\noindent
Now that we have shown that the extension axioms hold in $\mathcal{F}_{KR}$  
as well as in almost all finite strict partial orders, we have enough background to be able to prove the modal zero-one law with respect to the class of finite irreflexive transitive frames corresponding to provability logic.

\section{Validity in almost all finite irreflexive transitive frames}
\label{Frames}

Take $\Phi=\{p_1,\ldots, p_k\}$ or $\Phi=\{p_i \mid i\in \mathbb{N}\}$. The axiom system $\mathbf{AX^{\Phi,F}_{GL}}$ corresponding to validity in almost all finite frames of provability logic has the same axioms and rules as {\bf GL}, plus the following axiom schemas, for all $k\in \mathbb{N}$, where all $\varphi_i \in L(\Phi)$:
\begin{align}
\tag{T3} 
&\Box\Box\Box\bot \label{Threelayers}\\
\tag{DIAMOND-k}
& \Diamond \Diamond \top \wedge \bigwedge_{i\leq k} \Diamond(\Diamond \top \wedge \Box \varphi_i )\rightarrow \Box (\Diamond \top \rightarrow \Diamond ( \bigwedge_{i\leq k} \varphi_i)) \label{Diamond}\\
\tag{UMBRELLA-k}
& \Diamond\Diamond \top \wedge  \bigwedge_{i\leq k}\Diamond (\Box\bot \wedge \varphi_i ) \rightarrow \Diamond (\bigwedge_{i\leq k} \Diamond \varphi_i)\label{Umbrella}
\end{align}

\noindent
Here, UMBRELLA-0 is the formula $\Diamond\Diamond \top \wedge  \Diamond (\Box\bot\wedge  \varphi_0) \rightarrow \Diamond \Diamond \varphi_0$, which represents the property that direct successors of bottom layer worlds are never endpoints but have at least one successor in the top layer.

The formula DIAMOND-0 has been inspired by the well-known axiom $\Diamond\Box\varphi \rightarrow \Box\Diamond\varphi$ that characterizes confluence, also known as the diamond property: for all $x,y,z$, if $xRy$ and $xRz$, then there is a $w$ such that $yRw$ and $zRw$.

Note that in contrast to the theory $\mathbf{AX^{\Phi,M}_{GL}}$ introduced in Section~\ref{GL-models}, the axiom system $\mathbf{AX^{\Phi,F}_{GL}}$  gives
a normal modal logic, closed under uniform substitution.

Also notice that $\mathbf{AX^{\Phi,F}_{GL}}$ is given by an infinite set of axioms. It turns out that if we base our logic on an infinite set of atoms $\Phi=\{p_i \mid i \in  \mathbb{N}\}$, then for each $k\in  \mathbb{N}$, DIAMOND-k+1 and UMBRELLA-k+1 are strictly stronger than DIAMOND-k and UMBRELLA-k, respectively. So we have two infinite sets of axioms that both strictly increase in strength, thus by a classical result of Tarski, the modal theory generated by $\mathbf{AX^{\Phi,F}_{GL}}$ is not finitely axiomatizable.




For the proof of the zero-one law for frame validity, we will again need a completeness proof, this time of $\mathbf{AX^{\Phi,F}_{GL}}$ with respect to almost sure frame validity, including use of Lindenbaum's lemma 
and finitely many maximal $\mathbf{AX^{\Phi,F}_{GL}}$-consistent sets of formulas, each intersected with a finite  set of relevant formulas $\Lambda$. 

Below, we will define the {\em closure} of a sentence $\varphi\in L(\Phi)$. We may view
this closure as the set of formulas that are relevant  for making a (finite) 
countermodel against $\varphi$. 

\begin{definition}[Closure of a formula]
\label{Closure}
The {\em closure} of $\varphi$ with respect to $\mathbf{AX^{\Phi,F}_{GL}}$ is the
minimal set  $\Lambda$ of $\mathbf{AX^{\Phi,F}_{GL}}$-formulas such that:

\begin{enumerate}
\item $\varphi\in\Lambda$.
\item $\Box\Box\Box\bot\in\Lambda$.
\item If $\psi\in\Lambda$ and $\chi$ is a subformula of $\psi$, then
$\chi\in\Lambda$.
\item If $\psi\in\Lambda$ and $\psi$ itself is not a negation, then
$\neg\psi\in\Lambda$.
\item If $\Diamond \psi\in\Lambda$ and $\psi$ itself is not of the form $\Diamond\xi$ or $\neg\Box\xi$, 
then $\Diamond\Diamond\psi\in\Lambda$, and also $\Box\neg\psi$, $\Box\Box\neg\psi \in \Lambda$. 
\item If $\Box \psi\in\Lambda$ and $\psi$ itself is not of the form $\Box\xi$ or $\neg\Diamond\xi$, 
then $\Box\Box\psi\in\Lambda$, and also $\Diamond\neg\psi$, $\Diamond\Diamond\neg\psi \in \Lambda$. 

\end{enumerate}
\end{definition}

\noindent
Note that $\Lambda$ is a finite set of formulas, of size polynomial in the length of the formula $\varphi$ from which it is built.




\begin{definition}\label{precedence}
Let $\Lambda$ be a closure as defined above and let $\Delta, \Delta_1, \Delta_2$ be maximal $\mathbf{AX^{\Phi,F}_{GL}}$-consistent sets. We define:
\begin{itemize}
\item
 $\Delta^{\Lambda}:= \Delta \cap\Lambda$;
 \item $\Delta_1 \prec \Delta_2 $ iff for all $\Box\chi \in \Delta_1$, we have $\chi \in \Delta_2$;
 \item $\Delta_1^{\Lambda} \prec \Delta_2^{\Lambda}$ iff $\Delta_1 \prec \Delta_2$.
 \end{itemize}
\end{definition}



\begin{theorem}
\label{GL-trees}
For every formula $\varphi \in L(\Phi)$, the following are equivalent:
\begin{enumerate}
\item $\mathbf{AX^{\Phi,F}_{GL}} \vdash \varphi$;
\item $ \mathcal{F}_{KR}\models \varphi$;
\item $\lim_{n\to\infty} \mu_{n,\Phi}(\varphi) = 1$;
\item $\lim_{n\to\infty} \mu_{n,\Phi}(\varphi) \not = 0$.
\end{enumerate}
\end{theorem}

\begin{proof} We show a circle of implications. Let $\varphi\in L(\Phi)$.\\

\noindent
{\bf 1 $\Rightarrow$ 2}\\ 
  Suppose $\mathbf{AX^{\Phi,F}_{GL}} \vdash \varphi$. Because finite irreflexive Kleitman-Rothschild frames are finite strict partial orders that have no chains of length greater than $3$, the axioms and theorems of {\bf GL} + $\Box\Box\Box\bot$ hold in all Kleitman-Rothschild frames, therefore they are valid in $\mathcal{F}_{KR}$. 

So we only need to check the validity of the DIAMOND-k and UMBRELLA-k axioms in $\mathcal{F}_{KR}$ for all $k\geq 0$.\\

\noindent
DIAMOND-k-1: Fix $k\geq 1$,  take sentences $\varphi_i\in L(\Phi)$ for $i=1,\ldots, k-1$ and let  $\varphi= \Diamond \Diamond \top \wedge \bigwedge_{i\leq k-1} \Diamond(\Diamond \top \wedge \Box \varphi_i)\rightarrow \Box (\Diamond \top \rightarrow \Diamond ( \bigwedge_{i\leq k-1} \varphi_i))$. By Proposition~\ref{Categorical}, we know that each of the extension axioms of the form (b)  holds in $ \mathcal{F}_{KR}$. 
We want to show that $\varphi$ is valid in $\mathcal{F}_{KR}$.

To this end, let $V$ be any valuation  on the set of labelled states $W$ of $\mathcal{F}_{KR}$ and let $M= (\mathcal{F}_{KR}, V)$. Now take an arbitrary $b\in W$ and suppose that $M, b \models \Diamond \Diamond \top \wedge \bigwedge_{i\leq k-1} \Diamond(\Diamond \top \wedge \Box \varphi_i)$. Then $b$ is in the bottom layer $L_1$ and there are  worlds $x_0,\ldots,x_{k-1}$ (not necessarily distinct) in the middle layer $L_2$ such that for all $i \leq k-1$, we have $b <  x_i$  and $M, x_i \models \Box \varphi_i$. Now take any $x_{k}$ in $L_2$ with $b < x_{k}$. Then, by the extension axiom (b), there is an element $z$ in the upper layer $L_3$ such that $\bigwedge_{i \leq k}   x_i < z $. Now for that $z$, we have that $M, z \models \bigwedge_{i \leq k-1} \varphi_i$. But then $M, x_{k} \models \Diamond( \bigwedge_{i\leq k-1} \varphi_i)$, so because $x_{k}$ is an arbitrary direct successor of $b$, we have $M, b\models \Box (\Diamond \top \rightarrow \Diamond ( \bigwedge_{i\leq k-1} \varphi_i))$. To conclude, \vspace{-0.1cm} \[M, b \models 
 \Diamond \Diamond \top \wedge \bigwedge_{i\leq k-1} \Diamond(\Diamond \top \wedge \Box \varphi_i) \rightarrow \Box (\Diamond \top \rightarrow \Diamond ( \bigwedge_{i\leq k-1} \varphi_i)),\] so because $b$ and $V$ were arbitrary, we have \vspace{-0.1cm} \[\mathcal{F}_{KR}\models  \Diamond \Diamond \top \wedge \bigwedge_{i\leq k-1} \Diamond(\Diamond \top \wedge \Box \varphi_i) \rightarrow \Box (\Diamond \top \rightarrow \Diamond ( \bigwedge_{i\leq k-1} \varphi_i)),\vspace{-0.2cm}\] as desired.\\
 
 \noindent
 UMBRELLA-k-1: Fix $k\geq 1$,  take sentences $\varphi_i\in L(\Phi)$ for $i=1,\ldots, k-1$ and let  $\varphi=  \Diamond\Diamond \top \wedge  \bigwedge_{i\leq k-1}\Diamond (\Box\bot \wedge \varphi_i )  \rightarrow \Diamond (\bigwedge_{i\leq k-1} \Diamond \varphi_i)$.
 By Proposition~\ref{Categorical}, 
 we know that each of the extension axioms of the form (c)  holds in $\mathcal{F}_{KR}$. 
 We want to show that $\varphi$ is valid in $\mathcal{F}_{KR}$.

To this end, let  $V$ be any valuation  on the set of labelled states  $W$ 
of $\mathcal{F}_{KR}$ and let $M= (\mathcal{F}_{KR}, V)$. Now take an arbitrary  $b\in W$ and suppose that $M, b \models  \Diamond\Diamond \top \wedge  \bigwedge_{i\leq k-1}\Diamond (\Box\bot \wedge \varphi_i )$. Then $b$ is in the bottom layer $L_1$ and there are accessible worlds $x_0,\ldots, x_{k-1}$ (not necessarily distinct) in the upper layer $L_3$ such that for all $i \leq k-1$, we have $b < x_i$ and  $M, x_i\models \varphi_i$. By the extension axiom (c) from Definition~\ref{Extension}, there is an element $z$ in the middle layer $L_2$ such that $b < z$ and for all $i\leq k-1$, $z < x_i$. But that means that $M,z \models \bigwedge_{i\leq k-1} \Diamond \varphi_i$, therefore $M,b \models \Diamond (\bigwedge_{i\leq k-1} \Diamond \varphi_i)$. In conclusion, \vspace{-0.1cm} \[M,b \models  \Diamond\Diamond \top \wedge  \bigwedge_{i\leq k-1}\Diamond (\Box\bot \wedge \varphi_i )  \rightarrow \Diamond (\bigwedge_{i\leq k-1} \Diamond \varphi_i),\] so because $b$ and $V$ were arbitrary, we have \vspace{-0.1cm} \[\mathcal{F}_{KR}\models  \Diamond\Diamond \top \wedge  \bigwedge_{i\leq k-1}\Diamond (\Box\bot \wedge \varphi_i) \rightarrow \Diamond (\bigwedge_{i\leq k-1} \Diamond \varphi_i),\vspace{-0.2cm}\] as desired.\\

\noindent
{\bf 2 $\Rightarrow$ 3}\\ Suppose $\mathcal{F}_{KR}\models\varphi$. Using Van Benthem's translation (see Subsection~\ref{Benthem}), we can translate this as a $\Pi^1_1$ sentence being true in $\mathcal{F}_{KR}$ (viewed as model of the relevant second-order language): Universally quantify over predicates corresponding to all propositional atoms occurring in $\varphi$, to get a sentence of the form $\chi:=\forall P_1,\ldots,P_n \; \forall x \varphi^\ast$, where $\forall x \varphi^\ast$ is a first-order sentence. The claim is that $\forall x\varphi^\ast$  follows from a finite set of the extension axioms. 
Following the procedure of Kolaitis and Vardi~\cite[Lemma 4]{kolaitis1987}, we can prove that if every finite subset of $T_\mathrm{as-irr}\cup\{\exists x\neg\varphi^\ast(P_1,\ldots,P_n)\}$ is satisfiable over the extended vocabulary with $P_1, \ldots, P_n$, then by compactness, $T_\mathrm{as-irr}\cup\{\exists x\neg\varphi^\ast(P_1, \ldots,P_n)\}$ has a countable model over the extended vocabulary. The reduct of this model to the old language is still a countable model of $T_\mathrm{as-irr}$, and is therefore isomorphic to $\mathcal{F}_{KR}$ (by Proposition~\ref{Categorical}). But then $\mathcal{F}_{KR}\models\exists P_1,\ldots,P_n\exists x\neg\varphi^\ast$,  a contradiction. \\

\noindent
{\bf 3 $\Rightarrow$ 4}\\ Obvious, because $0 \not = 1$.\\


\noindent
{\bf 4 $\Rightarrow$ 1}\\
By contraposition. Let $\varphi\in L(\Phi)$ and suppose that $\mathbf{AX^{\Phi,F}_{GL}} \not \vdash \varphi$. Then $\neg \varphi$ is $\mathbf{AX^{\Phi,F}_{GL}}$-consistent. We will do a completeness proof by the finite step-by-step method 
(see, for example,~\cite{burgess1984,JonghVV}), but based on infinite maximal consistent sets, each of which is intersected with the same finite set of relevant formulas $\Lambda$, so that the constructed counter-model remains finite (see~\cite{Verbrugge1988}, \cite[footnote 3]{Joosten2020}).

In the following, we are first going to  construct a model
$M_{\varphi}=(W, R, V)$ that will 
contain a world where $\neg \varphi$ holds (Step 4 $\Rightarrow$ 1 (a)). Then we will embed this model into Kleitman-Rothschild frames of any large enough size to show that $\lim_{n\to\infty} \mu_{n,\Phi}(\varphi) = 0$ (Step 4 $\Rightarrow$ 1 (b)).\\

\noindent
{\bf Step 4 $\Rightarrow$ 1 (a)}\\
\noindent
By the Lindenbaum Lemma, 
we can extend $\{\neg \varphi\}$ to a maximal $\mathbf{AX^{\Phi,F}_{GL}}$-consistent set $\Psi$. Now define $\Psi^{\Lambda}:= \Psi \cap \Lambda$, where $\Lambda$ is as in Definition~\ref{Closure}, and introduce a world $s_\Psi$ corresponding to $\Psi^{\Lambda}$. 




We distinguish three cases for the step-by-step construction: {\bf U} (upper layer), {\bf M} (middle layer), and {\bf B} (bottom layer).\\

\noindent
{\bf Case U, with } {\boldmath$\Box\bot\in\Psi^{\Lambda}$}:\\ In this case we are done: a one-point counter-model $M_\varphi=(W,R,V)$ with $W=\{s_\Psi\}$, empty $R$, and for all $p\in\Phi$, $V_{s_\Psi}(p)$ iff $p\in\Psi^\Lambda$ suffices. \\

\noindent 
{\bf Case M, with } {\boldmath$\Box\bot\not\in \Psi^{\Lambda}$}, {\boldmath$\Box\Box\bot\in \Psi^{\Lambda}$}:\\ Let $\Diamond \psi_1,\ldots,\Diamond \psi_n$ be an enumeration of all the formulas of the form $\Diamond \psi$ in $\Psi^{\Lambda}$. Note that for all these formulas, $\Diamond \Diamond \psi_i \not\in \Psi^{\Lambda}$, because $\Box\Box\bot\in \Psi^{\Lambda}$. Take an arbitrary one of the $\psi_i$ for which $\Diamond\psi_i \in \Psi^{\Lambda}$. Claim: the set 
\[\Delta_i:=\{\Box \chi, \chi \mid \Box \chi \in \Psi  \} \cup \{\psi_i, \Box \neg \psi_i\}\] is $\mathbf{AX^{\Phi,F}_{GL}}$-consistent. For if not, then
 \[\{\Box \chi, \chi \mid \Box \chi \in \Psi \} \vdash_{\mathbf{AX^{\Phi,F}_{GL}}} \Box\neg\psi_i \rightarrow \neg\psi_i.\] 
 Because proofs are finite, there is a finite set $\chi_1,\ldots,\chi_k$ with $\Box\chi_1,\ldots \Box\chi_k\in\Psi$ and
  \[\{\Box \chi_j, \chi_j \mid j \in \{1,\ldots, k\} \} \vdash_{\mathbf{AX^{\Phi,F}_{GL}}} \Box\neg\psi_i \rightarrow \neg\psi_i.\] 
 Using necessitation, we get 
 \[\{\Box\Box \chi_j, \Box \chi_j \mid  j \in \{1,\ldots, k\} \} \vdash_{\mathbf{AX^{\Phi,F}_{GL}}} \Box(\Box\neg\psi_i \rightarrow \neg\psi_i).\]
Because we have $\vdash_{\mathbf{AX^{\Phi,F}_{GL}}} \Box\chi_j \rightarrow \Box\Box \chi_j$ for all  $j =1,\ldots, k$  and $\vdash_{\mathbf{AX^{\Phi,F}_{GL}}} \Box(\Box\neg\psi_i \rightarrow \neg\psi_i)\rightarrow \Box\neg \psi_i$, we can conclude:
 \[\{\Box \chi \mid \Box \chi \in \Psi \} \vdash_{\mathbf{AX^{\Phi,F}_{GL}}} \Box \neg\psi_i.\] Using 
 Proposition~\ref{maxconsistent}(4) and the fact that $\Box\neg\psi_i\in\Lambda$, this leads to $\Box\neg\psi_i\in \Psi^{\Lambda}$, contradicting our assumption that $\Diamond \psi_i\in \Psi^{\Lambda}$.
Also note that because $\Box\Box\bot\in \Psi$, by definition, $\Box\bot\in \Delta_i$.
We can now extend $\Delta_i$ to a maximal $\mathbf{AX^{\Phi,F}_{GL}}$-consistent set $\Psi_i$ by the Lindenbaum Lemma, and we define for each $i \in \{1, \ldots, n\}$ the set $\Psi_i^{\Lambda}:= \Psi_i \cap\Lambda$ and a world $s_{\Psi_i}$ corresponding to it. 
We have for all $i\in \{1, \ldots, n\}$ that $\Psi^{\Lambda} \prec \Psi_i^{\Lambda}$ as well as $\psi_i, \Box \neg\psi_i \in \Psi_i^{\Lambda}$.\\ 

\noindent
We have now finished creating a two-layer counter-model 
$M_{\varphi}= (W, R, V)$, which has:

\begin{itemize}
\item $W = \{s_{\Psi},s_{\Psi_1},\ldots,s_{\Psi_n}\}$; 
\item $R= \{\langle s_{\Psi}, s_{\Psi_i}\rangle \mid i\in\{1,\ldots,n\}\}$;
\item For each $p\in\Phi$ and $s_{\Gamma}\in W$: $V_{s_{\Gamma}}(p)=1$ iff $p\in \Gamma^{\Lambda}$.
\end{itemize}

\noindent
A truth lemma can be proved as in Case B below (but easier).\\

\noindent
{\bf Case B, with } {\boldmath$\Box\Box\bot\not\in \Psi^{\Lambda}$}:\\ In this case, we also look at all formulas of the form $\Diamond\psi \in \Psi^{\Lambda}$. 
We first divide this into two sets, as follows:
\begin{enumerate}
\item The set of $\Diamond$-formulas in $\Psi^{\Lambda}$ for which we have that $\Diamond \xi_{k+1}, \dots,\Diamond\xi_l \in \Psi^{\Lambda}$ but $\Diamond\Diamond \xi_{k+1}, \dots,\Diamond\Diamond\xi_l\not \in \Psi^{\Lambda}$ for some $l\in \mathbb{N}$, so $\Box\Box\neg \xi_{k+1}, \dots,\Box\Box\neg\xi_l \in \Psi^{\Lambda}$.\footnote{The formulas of the form $\Box\Box\neg \xi_{j}$ are in $\Lambda$ because of  Def.~\ref{Closure}, clause 5.}
\item The set of $\Diamond\Diamond$-formulas with $\Diamond\Diamond \xi_1, \dots,\Diamond\Diamond\xi_k\in  \Psi^{\Lambda}$.\\ Note that for these formulas, we also have $\Diamond \xi_1, \dots,\Diamond\xi_k\in \Psi^{\Lambda}$, because $GL \vdash \Diamond\Diamond \xi_i \rightarrow \Diamond \xi_i$. We will treat these pairs $\Diamond\Diamond\xi_i, \Diamond\xi_i$ for $i=1, \ldots,k$  at the same go.
\end{enumerate}

\noindent
Note that (1) and (2) lead to disjoint sets which together exhaust the $\Diamond$-formulas in $\Psi^{\Lambda}$. Altogether, that set now contains $\{\Diamond \xi_1, \ldots, \Diamond\xi_k, \Diamond\Diamond \xi_1, \ldots, \Diamond \Diamond\xi_k, \Diamond \xi_{k+1}, \dots,\Diamond\xi_l  \}$.\\

\noindent
Let us first check the formulas of type (1): $\Diamond \xi_{k+1},\ldots, \Diamond \xi_l\in\Psi^{\Lambda}$, but $\Box \Box \neg\xi_{k+1}, \ldots, \Box \Box \neg\xi_{l} \in \Psi^{\Lambda}$. We can now show by similar reasoning as in Case M that for each $i\in\{k+1,\ldots, l\}$, $\Delta_i= \{\Box \chi, \chi \mid \Box \chi \in \Psi \} \cup \{ \xi_i, \Box\neg \xi_i\}$ is $\mathbf{AX^{\Phi,F}_{GL}}$-consistent, so we can extend them to maximal $\mathbf{AX^{\Phi,F}_{GL}}$-consistent sets $\Psi_i$ 
and define $\Psi_i^{\Lambda}:=\Psi_i \cap \Lambda$ with 
$\Psi^{\Lambda} \prec \Psi_i^{\Lambda}$, and corresponding worlds $s_{\Psi_i}$ for all $i\in\{k+1, \ldots, l\}$.

We now claim that for all $i\in \{k+1,\ldots, l\}$, the world $s_{\Psi_i}$ is not in the top layer of the model with root  $s_\Psi$. 
To derive a contradiction, suppose that it is in the top layer, so  $\Box \bot \in \Psi_i^{\Lambda}$. Then 
also  $\Box\bot \wedge \xi_i \in \Psi_i$, 
so because 
$\Psi\prec \Psi_i$, we have $\Diamond (\Box\bot \wedge \xi_i) \in \Psi$. 
By UMBRELLA-0, we know 
that \vspace{-0.1cm}\[\vdash_{\mathbf{AX^{\Phi,F}_{GL}}} \Diamond \Diamond \top \wedge \Diamond (\Box \bot  \wedge \xi_i) \rightarrow \Diamond\Diamond\xi_i.\vspace{-0.1cm}\] Also having $\Diamond\Diamond\top\in\Psi$, we can now use Proposition~\ref{maxconsistent}(4) to conclude that  $\Diamond\Diamond \xi_i\in \Psi$. Therefore, because  $\Diamond\Diamond \xi_i\in \Lambda$, we also have $\Diamond\Diamond \xi_i\in \Psi^{\Lambda}$, contradicting our starting assumption that $\Diamond \xi_i$ is a type (1) formula. We conclude that $\Box \bot \not\in \Psi_i^{\Lambda}$, therefore,  $s_{\Psi_i}$ is in the middle layer. 

Let us now look for each of these $s_{\Psi_i}$ with $i$ in $k+1,\ldots, l$, which direct successors in the top layer they require. Any formulas of the form $\Diamond\chi\in\Psi_i^{\Lambda}$ have to be among the formulas $\Diamond\xi_1,\ldots,\Diamond\xi_k$ of type (2), for which $\Diamond\Diamond\xi_1,\ldots,\Diamond\Diamond\xi_k\in\Psi^{\Lambda}$.
Suppose $\Diamond\xi_j\in \Psi_i$ for some $j$ in $1,\ldots, k$ and $i$ in $k+1,\ldots, l$. Then we can show (just like in Case M) that there is a maximal consistent set $X_{i, j}$ with $\Psi_i \prec X_{i, j}$ and $\xi_j,\Box\bot \in X_{i, j}^{\Lambda}$. The world in the top layer corresponding to $X_{i,j}^{\Lambda}$ will be called $s_{X_{i, j}}$. Because $X_{i, j}^{\Lambda}$ is finite, we can describe it by $\Box\bot$ and a finite conjunction of literals, which we represent as $\chi_{i, j}$. For ease of reference in the next step, let us define:\\
 
 \noindent
$A:=\{\langle i,j\rangle\!\mid\!\Diamond\xi_j\in\!\Psi_i\!\mbox{ with }\!i\in\!k+1,\ldots,l\mbox{ and }
j\in\!1,\ldots,k\}$\\

\noindent
For the formulas of type (2), we have $\Diamond\Diamond\xi_i\in \Psi^{\Lambda}$. Moreover, we have for each $i\in \{1,\ldots, k\}$:\vspace{-0.1cm} \[GL+\Box\Box\Box\bot\vdash \Diamond\Diamond\xi_i \rightarrow \Diamond(\Box\bot \wedge \xi_i).\]
Therefore, by maximal $\mathbf{AX^{\Phi,F}_{GL}}$-consistency of $\Psi$, 
we have by Proposition~\ref{maxconsistent} that $\Diamond(\Box\bot \wedge \xi_i) \in \Psi$ for each $i\in \{1,\ldots, k\}$. 
We also have $\Diamond\Diamond \top\in\Psi$. UMBRELLA-k now gives us \[ \Psi\vdash_{\mathbf{AX^{\Phi,F}_{GL}}} \Diamond \Diamond \top \wedge \bigwedge_{i= 1,\ldots, k} \Diamond(\Box\bot \wedge \xi_i ) \rightarrow\]
\[ \Diamond ( \bigwedge_{i= 1,\ldots, k}\Diamond \xi_i)\]
  We conclude from maximal $\mathbf{AX^{\Phi,F}_{GL}}$-consistency of $\Psi$ and Proposition~\ref{maxconsistent}(4) that $\Diamond ( \bigwedge_{i= 1,\ldots, k} \Diamond \xi_i)\in \Psi$.

This means that we can construct {\em one} direct successor of $\Psi^{\Lambda}$ containing all the  $\Diamond \xi_i$ for $i\in \{1,\ldots, k\}$. To this end, let 
\[\Delta_1:= \{\Box \chi, \chi \mid \Box \chi \in \Psi \} \cup 
\{\Diamond \xi_1, \ldots, \Diamond \xi_k\} \]

\noindent
Claim: $\Delta_1$ is $\mathbf{AX^{\Phi,F}_{GL}}$-consistent. For if not, we would have:
\[\{\Box \chi, \chi \mid \Box \chi \in \Psi  \} \vdash_{\mathbf{AX^{\Phi,F}_{GL}}} \neg ( \bigwedge_{i= 1,\ldots, k} \Diamond \xi_i)\]
But then by the same reasoning as we used before (``boxing both sides" and using $GL\vdash \Box \chi \rightarrow \Box\Box\chi$) we conclude that
\[\{\Box \chi \mid \Box \chi \in \Psi  \} \vdash_{\mathbf{AX^{\Phi,F}_{GL}}} \Box\neg  (\bigwedge_{i= 1,\ldots, k} \Diamond \xi_i).\]
This contradicts $\Diamond (\bigwedge_{i= 1,\ldots, k}\Diamond\xi_i)\!\in\!\Psi$, which we showed above. Now that we know $\Delta_1$ to be $\mathbf{AX^{\Phi,F}_{GL}}$-consistent, 
we can extend it by the Lindenbaum Lemma 
to a maximal $\mathbf{AX^{\Phi,F}_{GL}}$-consistent set, 
 which we call $\Psi_1 \supseteq \Delta_1$. Let  $s_{\Psi_1}$ be the world corresponding to $\Psi_1^{\Lambda}$, with 
 $\Psi^{\Lambda}\prec \Psi_1^{\Lambda}$. 

Now we can use the same method as in Case M to find the required 
direct successors of $\Psi_1^{\Lambda}$. Namely, for all $i\in\{1,\ldots,k\}$ we find maximal $\mathbf{AX^{\Phi,F}_{GL}}$-consistent sets $\Xi_i$ 
and let $s_{\Xi_i}$ be the worlds corresponding to the $\Xi_i^{\Lambda}$, with  $\Psi_1^{\Lambda}\prec \Xi_i^{\Lambda}$ and $\xi_i\in \Xi_i$.\\


\noindent
We have now handled making direct successors of $\Psi^{\Lambda}$ for all the formulas of type (1) and type (2).
We can then finish off the step-by-step construction for Case B by populating the upper layer U using one appropriate restriction to $\Lambda$ of a maximal consistent set $\Xi_0$, 
as follows. We note that  $\Box\neg \xi_{i}\in \Psi_{i}^{\Lambda}$ for $i$ in $k+1, \ldots, l$, and that $\Box\Box\bot\in \Psi_1^{\Lambda}$. Let us take the following instance of the DIAMOND-(l-k) axiom scheme:
\[\Diamond \Diamond \top \wedge \bigwedge_{i\in \{k+1,\ldots, l\}} \Diamond(\Diamond \top \wedge \Box \neg \xi_i )\rightarrow \]
\[  \Box (\Diamond \top \rightarrow \Diamond ( \bigwedge_{i\in \{k+1,\ldots, l\}}  \neg \xi_i))\]
Now we have $\Diamond \Diamond \top \in \Psi^{\Lambda}$. Because $\Psi \prec \Psi_i$ and $\Diamond\top \wedge \Box\neg\xi_i\in \Psi_i$ for all $i$ in $k+1, \ldots, l$, we derive that

\[\bigwedge_{i\in \{k+1,\ldots, l\}} \Diamond(\Diamond \top \wedge \Box \neg \xi_i )\in \Psi.\]
  Now  
  by  one more application of Proposition~\ref{maxconsistent}(4), we have \[\Box (\Diamond \top \rightarrow \Diamond ( \bigwedge_{i\in \{k+1,\ldots, l\}}  \neg \xi_i))\in \Psi.\] 
   Because $\Psi \prec \Psi_j$ and $\Diamond\top\in\Psi_j$ for all $j$ in $1, k+1, \ldots, l$, we conclude 
   that  \[\Diamond \top \rightarrow \Diamond ( \bigwedge_{i\in \{k+1,\ldots, l\}}  \neg \xi_i)\in \Psi_j \mbox { for all } j\in\{1,k+1, \ldots, l\}.\] 
Now we can find one world $s_{\Xi_0}$ corresponding to $\Xi_0^{\Lambda}$ such that for all $j$ in $1, k+1,\ldots, l$, we have $\Psi_j^{\Lambda} \prec \Xi_0^{\Lambda}$. And moreover,  $\neg\xi_i  \in \Xi_0^{\Lambda}$ for all $i$ in $k+1,\ldots, l$.\\ 

\noindent
We have now finished creating our finite counter-model $M_\varphi= (W, R, V)$, which has:

\begin{itemize}
\item $W = \{s_{\Psi}\}\cup \{s_{\Psi_1}, s_{\Psi_{k+1}},\ldots, s_{\Psi_{l}}\} \cup$\\
 $\mbox{ }\hspace{0.7cm}\{s_{X_{i,j}} \mid \langle i, j \rangle \in A\}\cup \{s_{\Xi_i} \mid i \in \{1,\ldots,k\}\} \cup\{s_{\Xi_0}\}$
\item $R$ is the transitive closure of:\\ $\{\langle s_{\Psi},s_{\Psi_i}\rangle\mid i\in\{1,k+1,\ldots,l\}\}\cup$\\
 $\{\langle s_{\Psi_i},s_{X_{i,j}}\rangle \mid \langle i, j \rangle \in A\}\cup$\\
 $\{\langle s_{\Psi_1},s_{\Xi_i}\rangle \mid i \in \{1,\ldots,k\}\} \cup $\\
 $\{\langle s_{\Psi_i}, s_{\Xi_0}\rangle\mid  i\in\{1,k+1,\ldots,l\}\}$;
\item For each $p\in\Phi$ and $s_\Gamma \in W: V_{s_{\Gamma}}(p)= 1 \mbox{ iff } p\in \Gamma^{\Lambda}$.
\end{itemize}







\noindent
Now 
we can relatively easily prove a truth lemma, restricted to formulas from $\Lambda$, as follows.\\

\noindent
{\bf Truth Lemma}

\noindent
For all $\psi$ in $\Lambda$ and all 
worlds $s_\Gamma$ in $W$:\\ 
$M_{\varphi},s_\Gamma \models \psi$ iff $\psi \in \Gamma^{\Lambda}$.\\

\noindent
{\bf Proof} By induction on the construction of the formula. For atoms $p\in\Lambda$, the fact that $M_{\varphi}, s_\Gamma\models p$ iff $p \in \Gamma^{\Lambda}$ follows by the definition of $V$.\\

\noindent
{\bf Induction Hypothesis}: Suppose for some arbitrary $\chi,\xi\in \Lambda$, we have that for {\em all} 
worlds $s_\Delta$ in $W$:\\ $M_\varphi, s_\Delta \models \chi$ iff $\chi\in \Delta^{\Lambda}$ and $M_\varphi, s_\Delta \models \xi$ iff $\xi\in \Delta^{\Lambda}$.\\

\noindent
{\bf Inductive step}:
\begin{itemize}
\item {\em Negation}: Suppose $\neg \chi\in \Lambda$. Now by the truth definition, $M_\varphi, s_\Gamma \models \neg \chi$ iff $M_\varphi, s_\Gamma \not\models \chi$. By the induction hypothesis, the latter is equivalent to $\chi\not\in \Gamma^{\Lambda}$. But this in turn is equivalent by Proposition~\ref{maxconsistent}(1) to $\neg \chi\in \Gamma^{\Lambda}$.
\item {\em Conjunction}: Suppose $\chi\wedge\xi\in \Lambda$. Now by the truth definition, $M_\varphi, s_\Gamma\models\chi\wedge\xi$ iff $M_\varphi, s_\Gamma \models \chi$ and $M_\varphi, s_\Gamma \models \chi$. By the induction hypothesis, the latter is equivalent to $\chi\in\Gamma^{\Lambda}$ and $\xi\in\Gamma^{\Lambda}$, which by Proposition~\ref{maxconsistent}(2) is equivalent to $\chi\wedge \xi\in \Gamma^{\Lambda}$. 


\item {\em Box}: Suppose $\Box\chi\in \Lambda$. We know by the induction hypothesis that for all 
sets $\Delta^{\Lambda}$ in $W$, $M_\varphi, s_\Delta\models \chi$ iff $\chi \in \Delta^{\Lambda}$. We want to show that $M_\varphi, s_\Gamma \models \Box\chi$ iff $\Box\chi \in \Gamma^{\Lambda}$. 

For one direction, suppose that $\Box\chi\in \Gamma^{\Lambda}$, then by definition of $R$, for all $s_\Delta$ with $s_\Gamma R s_\Delta$, we have $\Gamma\prec\Delta$ so $\chi\in\Delta^{\Lambda}$, so by induction hypothesis, for all these $s_\Delta$, $M_\varphi, s_\Delta\models \chi$. Therefore by the truth definition, $M_\varphi, s_\Gamma\models \Box\chi$.

For the other direction, suppose that $\Box\chi\in\Lambda$ but  $\Box\chi\not\in \Gamma^{\Lambda}$. 
Then (by Definition~\ref{Closure} and Proposition~\ref{maxconsistent}(4)), we have $\Diamond \neg \chi \in \Gamma^{\Lambda}$.\footnote{or, if $\chi$ is of the form $\neg\chi_1$, then $\Diamond\chi_1\in\Gamma$, with $\Diamond\chi_1$ logically equivalent to $\Diamond \neg\chi$; in that case we reason further with $\Diamond\chi_1$.} 
Then in the step-by-step construction, in Case {\bf M} or Case {\bf B}, we have constructed a maximal $\mathbf{AX^{\Phi,F}_{GL}}$-consistent set $\Xi$ 
with $\Gamma \prec \Xi$ and $s_\Gamma R s_\Xi$, with $\neg \chi\in \Xi$, thus $\neg \chi\in \Xi^{\Lambda}$. Now by the induction hypothesis, we have $M_\varphi, s_\Xi \not\models \chi$, so by the truth definition, $M_\varphi, s_\Gamma \not\models \Box\chi$.\\
\end{itemize}



\noindent
Finally, from the truth lemma and the fact above that $\neg \varphi \in \Psi^{\Lambda}$, we have $M_\varphi, s_\Psi \not \models \varphi$, so we have found our counter-model.\\

\noindent
{\bf Step 4 $\Rightarrow$ 1 (b)}\\
\noindent
Now we need to show that  $\lim_{n\to\infty} \mu_{n,\Phi}(\varphi) = 0$.\\

\noindent
We claim that almost surely for a sufficiently large finite Kleitman-Rothschild type frame $F' = (W', R')$ of three layers, there is a bisimulation relation $Z$ from $M_\varphi=(W,R,V)$ defined in the (a) part of this step to $F'$, such that the image is a generated subframe  $F''=(W'', R'')$ of $F'$, with $R''=R' \cap W''$. We will define a valuation $V'$ on $F'$ and define $V''$ as the restriction of $V'$ to $W''$ and such that for all $w\in W, w''\in W''$: If $wZw''$, then $w$ and $w''$ have the same valuation. 
Then, once the bisimulation $Z$ is given, suppose that $s_\Psi Zs''$ for some $s''\in W''$. By the bisimulation theorem~\cite{Benthem1983}, we have that for all $\psi\in L(\Phi)$,  $M_\varphi, s_\Psi\models\psi \Leftrightarrow M'',s''\models\psi$, in particular, $M'',s''\not\models\varphi$. Because $M''$ is a generated submodel of $M'=(W',R',V')$, we also have $M',s''\not\models\varphi$.  Conclusion: $\lim_{n\to\infty} \mu_{n,\Phi}(\varphi)=0$.\\

\noindent
We now sketch how to define the above-claimed bisimulation $Z$ from  $M_\varphi$ to a generated subframe of such a sufficiently large Kleitman-Rothschild frame $F' = (W', R')$. There are three cases, corresponding to {\rm Case U}, {\rm Case M}, and {\rm Case B} of the step-by-step construction of the counter-model $M_{\varphi}$ in Step 4 $\Rightarrow$ 1 (a). One by one, we will show that the constructed counter-model can almost surely be mapped by a bisimulation to a generated subframe of a Kleitman-Rothschild frame, as the number of nodes grows large enough.\\

\noindent
{\bf Case U}\\
The one-point counter-model $M_\varphi=(W,R,V)$ against $\varphi$, with $W=\{s_\Psi\}$, can be turned into a counter-model on every three-layer Kleitman-Rothschild frame $F'$ as follows. Take a world $u$ in the top layer of $F'$, define $Z$ by $s_\Psi Z u$ and take a valuation on $L(\Phi)$ that corresponds on that world $u$ with the valuation of world $s_\Psi$ in  $M_\varphi$. 
Then $Z$ is a bisimulation from   $M_\varphi$ to a model on a one-point generated subframe of $F'$. This world provides a counterexample showing $F'\not\models\varphi$.\\

\noindent
{\bf Case M}\\
The two-layer model $M_\varphi$ defined in Case M of part (a) of this step almost surely has as a bisimilar image a generated subframe of  a large enough Kleitman-Rothschild frame $F' =(W', R', V')$, as follows. Take a world $m$ in the middle layer of the Kleitman-Rothschild frame $F'$ with sufficiently many (at least $n$) successors in the top layer, say, $u_1, \ldots, u_m$ with $m\geq n$.  Take $F''=(W'', R'', V'')$ to be the generated upward-closed frame of $F'$ with root $m$. Define a mapping $Z$ by $s_\Psi Z m$, $s_{\Psi_i} Z u_i$ for all $i < n$, and $s_{\Psi_n} Z u_i$ for all $i$ in $n, \ldots, m$. This mapping satisfies the `forth' condition as well as the `back' condition.
Choose the valuation $V''$ on $L(\Phi)$ such that $m$ has the same valuation as $s_\Psi$, while $u_i$ has the same valuation as  $s_{\Psi_i}$ for $i<n$, and $u_i$ has the same valuation as $s_{\Psi_n}$ for all $i$ in $n, \ldots, m$. So $Z$ is a bisimulation.\\

\noindent
{\bf Case B}\\
The three-layer model $M_\varphi= (W, R, V)$ defined in Case B of part (a) of this step 
 can be embedded into almost every sufficiently large Kleitman-Rothschild frame $F'=(W',R')$ in the sense that there is a bisimulation to a model on a  generated subframe $F''$ of $F'$. 
Pick different pairwise distinct elements $u_{i,j}$ for all $\langle i,j\rangle\in A$ and $v_i$ for all $i\in\{1,\ldots,k\}$ in the upper layer $L_3$ of $F'$.\footnote{This is possible because there are at least $k \cdot (l-k) + k$  elements of $L_3$.}  
Now take any $b$ in bottom layer $L_1$. Then by a number of applications of extension axiom (c), we find members $m_1, m_{k+1}, \ldots, m_l$ in the middle layer $L_2$ such that:
 \begin{itemize}
 \item $bR' m_i$ for all $i$ in $1, k+1,\ldots, l$;
 \item $m_i R' u_{i,j}$ for all $\langle i,j\rangle\in A$;
 \item $m_1 R' v_i$ for all $i\in\{1,\ldots,k\}$;
 \item  {\em not} $m_1 R' u_{i,j}$ for any $\langle i,j\rangle\in A$;
 \item {\em not} $m_j R' v_i$ for any $j\in \{ k+1,\ldots, l\}$ and $i\in\{1,\ldots,k\}$;
 \item {\em not} $m_i R' u_{i',j}$ for any $i\in \{ k+1,\ldots, l\}$ and $\langle i',j\rangle\in A$ with $i \not = i'$;
 \end{itemize}
 Finally, by extension axiom (b), there is a $w_0$ in $L_3$ different from all $u_{i,j}$, $v_{i}$ such that $m_i R' w_0$ for all $i$ in $1, k+1, \ldots, l$. Let $F''=(W'',R'')$ be the upward-closed subframe generated by $b$, with $R''=R'\cap W''$.
 
 Now define mapping $Z$ from $M_\varphi$ to $F''$ such that:
 \begin{itemize}
 \item $s_{\Psi} Z b$;
 \item $s_{\Psi_i} Z m_i$ for $i \in \{k+1,\ldots, l\}$;
 \item $s_{X_{i,j}} Z u_{i,j}$ for all $\langle i,j\rangle\in A$;
 \item $s_{\Xi_{i}} Z v_{i}$ for all $i\in\{1,\ldots,k\}$;
\item $s_{\Xi_0} Z v$ for all other $v$ (not of the form $u_{i,j}$) with $m_i R' v$ for $i \in \{k+1,\ldots, l\}$;
\item $s_{\Xi_0} Z w_0$;
\item $s_{\Psi_1} Z m$ for all $m$ in $L_2$ other than $m_{k+1},\ldots, m_l$;
\item Divide the (many) $w\in L_3$ such that for all $i\in\{k+1,\ldots,l\}$ {\em not} $m_i R' w$,  randomly into a partition of about equal-sized subsets, one by one corresponding to $Z$- images of each of $s_{\Xi_0}$, $s_{\Xi_{i}}$ for  each $i\in\{1,\ldots,k\}$.
\end{itemize}
Now define the valuation $V''$ on $F''$ such that for all $p\in \Phi$ and all $s_\Gamma \in W$ and $s''\in W''$, if $s_\Gamma Z s''$, then $V''_{s''}(p)=1$ iff $V_{s_\Gamma}(p)=1$.
Finally, one can check that $Z$ also satisfies the two other conditions for bisimulations:
\begin{itemize}
\item {\bf Forth:} Suppose $s_\Gamma Z s''$ and $s_\Gamma R s_\Delta$. Then case by case, one can show that there is a $v''\in W''$ such that $s''R'' v''$ and $s_\Delta Zv''$; 
\item {\bf Back:} Suppose $s_\Gamma Zs''$ and $s''R'' v''$. Then case by case, one can show that there is an $s_\Delta\in W$ such that $s_\Gamma R s_\Delta$ and $s_\Delta Zv''$. 
\end{itemize}

 \vspace{0.25cm}
\noindent
Hereby we have sketched a proof that on almost all  large enough Kleitman Rothschild frame frames, $\varphi$ is not valid. \vspace{-0.1cm}\\

\noindent
To conclude, all of 1, 2, 3, and 4  are equivalent.
\end{proof}

\section{Complexity of almost sure model and frame satisfiability}
\label{Complexity}
It is well known that the satisfiability problem and the validity problem for {\bf GL} are PSPACE-complete (for a proof sketch, see~\cite{Verbrugge2017}), just like for other well-known modal logics such as {\bf K} and {\bf S4}. In contrast, for enumerably infinite vocabulary $\Phi$, the problem whether $\lim_{n\to\infty} \nu_{n,\Phi}(\varphi) = 0$  is in $\Delta^p_2$ (for the dag-representation of formulas), by adapting~\cite[Theorem 4.17]{halpern1994}. If $\Phi$ is finite, the decision problem whether $\lim_{n\to\infty} \nu_{n,\Phi}(\varphi) = 0$ is even in $P$, because you only need to check validity of $\varphi$ in the fixed finite canonical model $\mathrm{M}^{\Phi}_{GL}$. For example, for $\Phi=\{p_1,p_2\}$, this model contains only 16 worlds, see Figure 2.

The problem whether $\lim_{n\to\infty} \mu_{n,\Phi}(\varphi) = 0$ is in NP, more precisely, NP-complete for enumerably infinite vocabulary $\Phi$. To show that it is in NP, suppose you need to decide whether $\lim_{n\to\infty} \mu_{n,\Phi}(\varphi) = 0$. By the proof of part 4 $\Rightarrow$ 1 of Theorem~\ref{GL-trees}, you can simply guess an at most 3-level irreflexive transitive frame of the appropriate form and of size $< \mid\varphi\mid^3$, a model on it and a world in that model, and check (in polynomial time) whether $\varphi$ is not true in that world. NP-hardness is immediate for $\Phi$ infinite: for propositional $\psi$, we have $\psi\in\mathbf{SAT}$  iff  $\lim_{n\to\infty} \mu_{n,\Phi}(\psi) = 0$. 

In conclusion, if the polynomial hierarchy does not collapse and in particular (as most complexity theorists believe) $\Delta^p_2 \not=$ PSPACE and NP $\not=$ PSPACE, then the problems of deciding whether a formula is {\em almost always} valid in finite models or frames of provability logic are easier than deciding whether it is {\em always} valid. 
For comparison, remember that for first-order logic the difference between validity and almost sure validity is a lot starker still: Grandjean~\cite{grandjean1983} proved that the decidability problem of almost sure validity in the finite is only PSPACE-complete, while the validity problem on {\em all} structures is undecidable~\cite{church1936,turing1937} and the validity problem on all finite structures is not even recursively enumerable~\cite{Trakhtenbrot1950}.

\section{Conclusion and future work}
\label{Discussion}

We have proved zero-one laws for provability logic 
with respect to both model and frame validity. On the way, we have axiomatized validity in almost all relevant finite models and in almost all relevant finite frames, leading to two different axiom systems. If the polynomial hierarchy does not collapse, the two problems of `almost sure model/frame validity' are less complex than `validity in all models/frames'.  

Among finite frames in general, partial orders are pretty rare -- using Fagin's extension axioms, it is easy to show that almost all finite frames are {\em not} partial orders. Therefore, results about almost sure frame validities in the finite do not transfer between frames in general and strict partial orders. Indeed, the logic of frame validities on finite irreflexive partial orders studied here is quite different from the modal logic of the validities in almost all finite frames~\cite{goranko2003,Goranko2020}. 
One of the most interesting results in~\cite{goranko2003} is that frame validity does not transfer from almost all finite $\mathcal{K}$-frames to the countable random frame, although it does transfer in the other direction. In contrast, we have shown that for irreflexive transitive frames, validity does transfer in both directions between almost all finite frames and the countable random irreflexive Kleitman-Rothschild frame.

\subsection{Future work} Currently, we are proving similar 0-1 laws for logics of reflexive transitive frames, such as {\bf S4} and Grzegorczyk logic, axiomatizing both almost sure model validity and almost sure frame validity. It turns out that Halpern and Kapron's claim that there is a 0-1 law for $\mathcal{S}4$ frame validity can still be salvaged, albeit with a different, stronger axiom system, containing two infinite series of umbrella and diamond axioms similar to the ones in the current paper. Furthermore, it appears that one can do the same for logics of transitive frames that may be neither reflexive nor irreflexive, such as {\bf K4} and weak Grzegorczyk logic.

\bibliographystyle{IEEEtran}
\bibliography{aiml18}




\end{document}